\documentclass[hidelinks,onefignum,onetabnum]{siamart220329}

\usepackage{lipsum}
\usepackage{amsfonts}
\usepackage{graphicx}
\usepackage{epstopdf}
\usepackage{algorithmic}
\ifpdf
  \DeclareGraphicsExtensions{.eps,.pdf,.png,.jpg}
\else
  \DeclareGraphicsExtensions{.eps}
\fi

\newsiamremark{remark}{Remark}
\newsiamremark{hypothesis}{Hypothesis}
\crefname{hypothesis}{Hypothesis}{Hypotheses}
\newsiamthm{claim}{Claim}

\title{The sample complexity of sparse multi-reference alignment and single-particle cryo-electron microscopy}
\author{Tamir Bendory  and Dan Edidin}

\usepackage{amsopn}

\usepackage{mathtools} 
\usepackage{graphicx,multirow}
\usepackage{amssymb,latexsym,amsmath,mathrsfs}
\usepackage{color} 
\usepackage[all]{xy} 
\usepackage{caption}
\usepackage{subcaption}
\usepackage{url}
\numberwithin{equation}{section}

\renewcommand\L{{\mathbb L}}

\newcommand\Lf{{\mathbb L}_f}
\newcommand\I{\iota}
 \newcommand\C{\mathbb{C}}

\newcommand\E{\mathbb{E}}
 \renewcommand\P{\mathbb{P}}
 
\newcommand\R{\mathbb{R}} \newcommand\Z{\mathbb{Z}}
 
\newcommand \K{\mathbb{K}}
\newcommand\NN{\mathcal{N}} 
\newcommand\V{\mathcal{V}}

\newcommand\hg{\widehat{g}} 
 
\newcommand\hf{\widehat{f}}

\DeclareMathOperator\trace{trace}
\DeclareMathOperator\Mat{Mat}
\DeclareMathOperator\SO{SO}
 \DeclareMathOperator\rank{rank}
\DeclareMathOperator\Gr{Gr} \DeclareMathOperator\GL{GL}

\DeclareMathOperator\Span{span}

 \DeclareMathOperator{\Hom}{Hom}

 \newtheorem{example}[theorem]{Example}

\newcommand{\rev}[1]{{\color{black}{#1}}}

\begin{document}

\maketitle

% REQUIRED
\begin{abstract}
	Multi-reference alignment (MRA) is the problem of  recovering  a signal from its multiple noisy copies, each acted upon by a random group element.
	MRA is mainly motivated by single-particle cryo-electron microscopy (cryo-EM) that has recently joined X-ray crystallography as one of the two leading technologies to reconstruct biological molecular structures. 
	Previous papers have shown that in the high noise regime, the  sample complexity of MRA and cryo-EM is $n=\omega(\sigma^{2d})$, where $n$ is the number of observations, $\sigma^2$ is the variance of the noise, and $d$ is the lowest-order moment of the observations that uniquely determines the signal. In particular, it was shown that in many cases, $d=3$ for generic signals, and thus the sample complexity is $n=\omega(\sigma^6)$.
	
	In this paper, we analyze the second moment of the MRA and cryo-EM models. 
	First, we show that in both models the second moment determines the signal up to a set of unitary matrices, whose dimension is governed by the
	decomposition of the space of signals into irreducible representations
	of the group.
	Second, we  derive sparsity conditions under which a signal can be recovered from the second moment, implying  sample complexity of $n=\omega(\sigma^4)$.
	Notably, we show that the sample complexity of cryo-EM is $n=\omega(\sigma^4)$ if 
	at most one third of the coefficients representing the molecular structure are non-zero; this bound is near-optimal. The analysis is based on tools from representation theory and algebraic geometry. 
	We also derive bounds on recovering a sparse signal from its power spectrum, which  is the main computational problem of X-ray crystallography.
\end{abstract}

% REQUIRED
%\begin{keywords}
%example, \LaTeX
%\end{keywords}

% REQUIRED
%\begin{MSCcodes}
%68Q25, 68R10, 68U05
%\end{MSCcodes}

\section{Introduction}
This paper studies  the multi-reference alignment (MRA) model of estimating a signal from its multiple noisy copies, each acted upon by a random group element. 
Let~$G$ be  a compact group acting on \rev{an $N$-dimensional vector space~$V$ that can be identified with $\R^N$.}
Each MRA observation $y$ is drawn from 
\begin{eqnarray} \label{eq:mra}
	y = g\cdot f + \varepsilon,
\end{eqnarray}
where  $g\in G$, $\varepsilon\sim\NN(0,\sigma^2 I)$ is a Gaussian noise
\rev{vector} independent of $g$,    $\cdot$ denotes the group action, and $f\in V$.  
We assume that the distribution over $G$ is uniform (Haar).
The goal is to estimate the signal $f\in V$ from $n$ realizations 
\begin{equation}
	y_i = g_i\cdot f + \varepsilon_i\quad i=1,\ldots,n.
\end{equation}
%Evidently, it is impossible to distinguish between $f$ and $\tilde g \cdot f$ for any $\tilde g\in G$,   from the observations $y_1,\ldots,y_n$, with no prior knowledge on $f$. Thus, we can only hope to recover the orbit of $f\in V$ under $G$. 
\rev{Evidently, given a set of observations $y_1,\ldots,y_n$ and  with no prior knowledge on $f$, it is impossible to distinguish between $f$ and $\tilde g \cdot f$ for any $\tilde g\in G$}. Thus, we can only hope to recover the orbit of $f\in V$ under $G$.

A wide range of MRA models have been studied in recent years.
The simplest and most studied model is when a signal in $V=\R^N$ is estimated from its multiple circularly shifted, noisy copies, namely $G=\Z_\rev{N}$~\cite{bandeira2014multireference,bendory2017bispectrum,abbe2018multireference,bandeira202optimal,perry2019sample}. 
Figure~\ref{fig:1Dmra} illustrates observations drawn from this model.
Additional MRA models include the dihedral group acting on $\R^N$~\cite{bendory2022dihedral}, 
the group of two-dimensional rotations SO(2) acting on band-limited images~\cite{bandeira2020non,ma2019heterogeneous,janco2022accelerated}, the group of three-dimensional rotations SO(3) acting on band-limited signals on the sphere~\cite{bandeira2017estimation,liu2021algorithms},  as well as additional setups~\cite{pumir2021generalized,hirn2021wavelet,bendory2022compactification}.  
The results of this paper hold for any MRA model when a compact group $G$ is acting on a finite-dimensional space $V$; specific examples are provided in Section~\ref{sec:ambiguities_examples}.

\begin{figure}[h]
	\centering
	\includegraphics[width=0.5\linewidth]{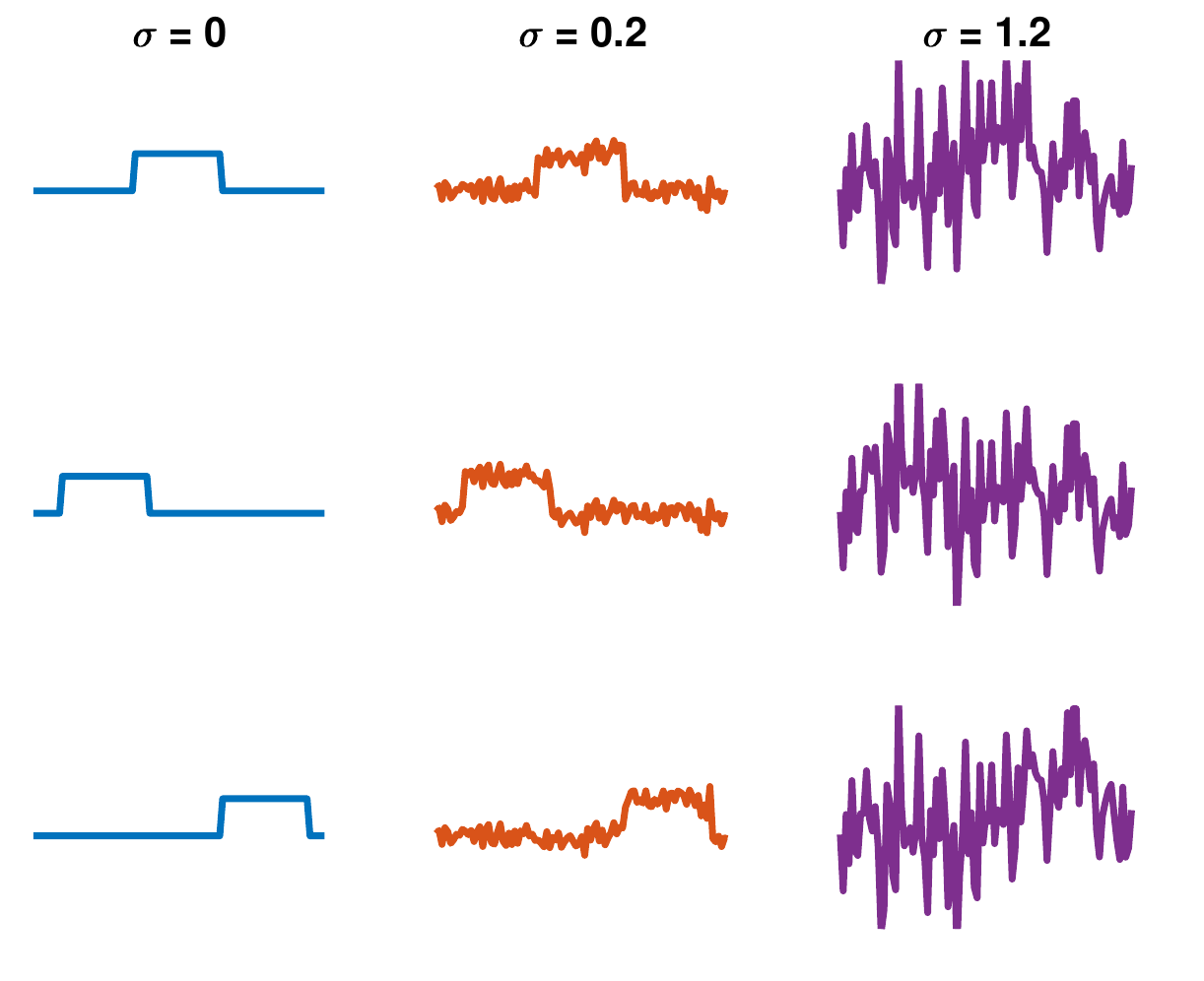}
	\caption{\label{fig:1Dmra} An example of the one-dimensional MRA setup, where a signal in $\R^N$ is acted upon by random elements of the group of circular shifts $\Z_N$. 
		The left column shows three shifted copies of the signal, corresponding to noiseless measurements (i.e., $\sigma=0$). In this case, all three observations are admissible solutions as the signal can be estimated only up to a group action. The middle and right columns present the same observations, with low noise level of $\sigma=0.2$ and high noise level of $\sigma=1.2$. This paper focuses on the extremely high noise level $\sigma\to\infty$ when the signal is swamped by noise. Figure credit:~\cite{bendory2020single}. }
\end{figure}

The MRA model is mainly motivated by single-particle cryo-electron microscopy (cryo-EM)---an increasingly popular technology that has joined X-ray crystallography as one of the two leading technologies to reconstruct molecular structures~\cite{frank2006three,nogales2016development}. 
Under some simplified assumptions, the cryo-EM generative model reads 
\begin{equation} \label{eq:cryoEM_model}
	y = T (g\cdot f) + \varepsilon, \quad  g\in G,
\end{equation}
where $G$ is the group of three-dimensional rotations SO(3), and $T$ is a tomographic projection  acting by 
\begin{equation} \label{eq:tomographic_projection}
	Tf(x_1,x_2)=\int_\R f(x_1,x_2,x_3)dx_3.
\end{equation}
The celebrated Fourier Slice Theorem states that the 2-D Fourier transform of a tomographic projection is equal to a 2-D slice of the  volume's 3-D Fourier transform~\cite{natterer2001mathematics}. 
This motivates analyzing  the cryo-EM model in Fourier space, which is indeed the common practice. 
Notably, the noise level in cryo-EM images is very high; Figure~\ref{fig:cryo} shows several experimental cryo-EM images.  
We refer the reader to recent surveys on the mathematical and algorithmic aspects of cryo-EM~\cite{singer2018mathematics,bendory2020single,singer2020computational}.

While the \rev{random linear} action of 3-D rotation followed by a tomographic projection 
does not constitute a group action,  we will show that the results of this paper apply to the cryo-EM model as well.
The emerging molecular reconstruction technology of X-ray free-electron lasers (XFEL) also obeys the model~\eqref{eq:cryoEM_model} with one important distinction: the phases in Fourier space are unavailable~\cite{spence2012x,maia2016trickle}.  

\begin{figure}[h!]
	\centering
	\begin{subfigure}[h]{0.45\textwidth}
		\centering
		\includegraphics[scale=0.8]{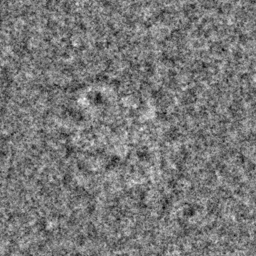}		
		\caption{\texttt{EMPIAR-10028}}
	\end{subfigure} \hfill 
	\begin{subfigure}[h]{0.45\textwidth}
		\centering
		\includegraphics[scale=0.8]{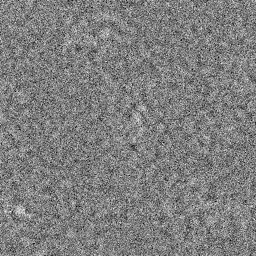}
		\caption{\texttt{EMPIAR-10073}}
	\end{subfigure} \hfill %\hspace{2pt}
	
	\begin{subfigure}[h]{0.45\textwidth}
		\centering
		\includegraphics[scale=0.8]{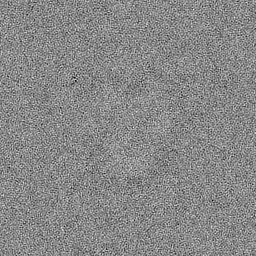}
		\caption{\texttt{EMPIAR-10081}}
	\end{subfigure} \hfill
	\begin{subfigure}[h]{0.45\textwidth}
		\centering
		\includegraphics[scale=0.8]{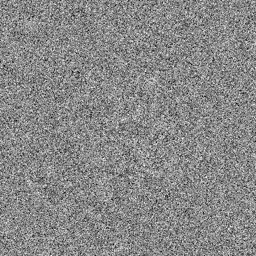}
		\caption{\texttt{EMPIAR-10061}}
	\end{subfigure}
	\caption{\label{fig:cryo} 
		A collection of cryo-EM experimental images, taken from the the Electron Microscopy Public Image Archive (EMPIAR) \texttt{https://www.ebi.ac.uk/empiar/}. 
		The corresponding molecular structures are available at the  Electron Microscopy Data Bank (EMDB) \texttt{https://www.ebi.ac.uk/emdb}.
		(a) \texttt{EMPIAR-10028} (corresponding entry \texttt{EMD-2660}): Plasmodium falciparum 80S ribosome bound to the anti-protozoan drug emetine~\cite{wong2014cryo};
		(b) \texttt{EMPIAR-10073}  (corresponding entry \texttt{EMD-8012}): yeast spliceosomal U4/U6.U5 tri-snRNP~\cite{nguyen2016cryo};
		(c) \texttt{EMPIAR-10081} (corresponding entry \texttt{EMD-8511}): human HCN1 hyperpolarization-activated cyclic nucleotide-gated ion channel~\cite{lee2017structures}; 
		(d) \texttt{EMPIAR-10061} (corresponding entry \texttt{EMD-2984}): beta-galactosidase in complex with a cell-permeant inhibitor~\cite{bartesaghi20152}.} 
\end{figure}

\paragraph{MRA analysis in the high and low noise regimes}
In the low noise regime, when the signal dominates the noise, the group elements $g_1,\ldots,g_n\in G$ can be usually estimated accurately from the observations, see for example~\cite{singer2011angular,boumal2016nonconvex,chen2018projected,rosen2019se,ling2022near}. 
If we denote the estimated group elements by $\hg_1,\ldots,\hg_n\in G$, then an estimator $\hf$ can be constructed  by applying the inverse group elements and averaging:
\begin{equation*}
	\hf = \frac{1}{n}\sum_{i=1}^n \hg_i^{-1}y_i.
\end{equation*}
In cryo-EM, while the statistical model is more involved~\eqref{eq:cryoEM_model}, the group elements can be estimated as well based on the common-lines geometrical property~\cite{singer2011three,shkolnisky2012viewing}, and thus recovering the molecular structure reduces to a linear inverse problem  for which many effective techniques exist~\cite{natterer2001mathematics}.

Motivated by cryo-EM, this work focuses on the high noise regime, when the signal is swamped by noise, and thus the group elements cannot be accurately estimated~\cite{bendory2018toward,aguerrebere2016fundamental,romanov2021multi}. 
Consequently, one needs to develop methods to estimate the signal $f$ directly, without estimating the group elements as an intermediate step.
In particular, two main estimation methods dominate the  MRA literature. 
The first is based on optimizing the marginalized likelihood function, using methods such as expectation-maximization~\cite{bendory2017bispectrum,ma2019heterogeneous,bendory2022super,bendory2022dihedral,janco2022accelerated,kreymer2022approximate,bendory2022sparse}. While these techniques are highly successful, and are the state-of-the-art methods in cryo-EM~\cite{sigworth1998maximum,scheres2012relion,punjani2017cryosparc}, their properties are currently not well-understood~\cite{fan2020likelihood,katsevich2020likelihood,brunel2019learning,fan2021maximum}.
The second approach is based on the method of moments---a classical parameter estimation technique, tracing back to the seminal paper of Pearson~\cite{pearson1894contributions}. 
In the method of moment\rev{s}, the idea is to 
find a signal which is consistent with the empirical moments (which are estimates of the population moments).
The method of moments was applied to a wide range of MRA models~\cite{bendory2017bispectrum,bandeira2017estimation,perry2019sample,abbe2018multireference,chen2018spectral,boumal2018heterogeneous,ma2019heterogeneous,pumir2021generalized,aizenbud2021rank,hirn2021wavelet,liu2021algorithms,bendory2022super,bendory2022dihedral,bendory2022sparse,ghosh2022sparse,abas2022generalized}, as well as  to construct {ab initio} models in cryo-EM~\cite{bhamre2015orthogonal,bhamre2017anisotropic,levin20183d,sharon2020method,bendory2018toward,lan2021random,huang2022orthogonal} and XFEL~\cite{saldin2010structure,donatelli2015iterative}. 
In this work, we focus on the method of moments due to its appealing statistical properties that are introduced next.

\paragraph{Sample complexity}
In the high noise regime $\sigma\to\infty$, when the dimension of the signal is finite, it was shown that  a necessary condition for recovery is $n= \omega(\sigma^{2d})$ (namely, $n/\sigma^{2d}\to\infty$ as $n,\sigma\to\infty$), where $d$ is the lowest-order moment that determines the orbit of the signal uniquely\footnote{This is not necessarily true when the dimension of the signal grows with the noise level and the number of observations~\cite{romanov2021multi,dou2022rates}.}~\cite{bandeira202optimal,abbe2018estimation,perry2019sample}.
\rev{Therefore, determining the sample complexity in the high noise regime reduces to analyzing moment equations.}
In~\cite{bandeira2017estimation, edidin2023orbit}, it was shown that 
in many cases, if the distribution of the group elements is uniform (as we assume in this paper), 
$d=3$ suffices to determine  almost all signals, implying sample complexity of  $n=\omega(\sigma^6)$; this is also true for cryo-EM.
Moreover, in some cases, an efficient algorithm to recover the signal at the optimal estimation rate was devised. 
For example, if $V\in\R^N$ and $G=\Z_N$, a generic signal can be recovered efficiently from the third moment, called the bispectrum, using a variety of efficient  algorithms~\cite{bendory2017bispectrum,perry2019sample}; see also~\cite{liu2021algorithms}.

We mention that when the distribution of the group elements is non-uniform, the MRA problem is usually easier, and signal recovery  may be possible from the second moment~\cite{abbe2018multireference,bendory2022dihedral,sharon2020method}.
In fact, uniform distribution can be thought of as the worst-case scenario of the MRA model~\eqref{eq:mra} since, no matter what the original distribution over the group elements is, one can force a uniform distribution by generating a new set of observations:
\begin{equation}
	z_i = \tilde g_i\cdot y_i = (\tilde g_ig_i)\cdot f + \tilde g_i\cdot \varepsilon_i,
\end{equation}
where $\tilde g_i$ is drawn from a uniform distribution (and thus  the distribution of $\tilde gg$ is also uniform). 
This is not necessarily true for the cryo-EM model.

\paragraph{Main contributions: Signal recovery from the second moment}
This work studies signal recovery from the second moment of the MRA observations: \begin{equation}
	\E yy^* = \int_{G} (g\cdot f)(g\cdot f)^*dg \rev{+\sigma^2I}.
\end{equation}
\rev{Since we assume to know $\sigma^2$, we henceforth omit the effect of the noise.}
If  we view $g\cdot f$ as a column vector, then $(g\cdot f)(g\cdot f)^*$ is a rank-one matrix, and thus the second moment is an integral over rank-one Hermitian matrices. 
Recall that the second moment can be estimated from samples 
\begin{equation}
	\E yy^* \approx \frac{1}{n}\sum_{i=1}^n y_iy_i^*.
\end{equation}
When $n=\omega(\sigma^4)$, \rev{$\frac{1}{n}\sum_{i=1}^n y_iy_i^*$ almost surely convergences to $\E yy^*$. }
 %$\E yy^* = \frac{1}{n}\sum_{i=1}^n y_iy_i^*$ almost surely.
 \rev{In this paper, we identify a class of signals that are determined uniquely by $\E yy^*$. This in turn implies that the sample complexity of the problem, for this class of signals, is  $n=\omega(\sigma^4)$ and not $n=\omega(\sigma^6)$ as for generic signals~\cite{bandeira2017estimation,bendory2017bispectrum,perry2019sample}.}

The first contribution of this paper, introduced in Section~\ref{sec:second_moment}, is \rev{a precise characterization of}  the set of signals 
\rev{having the same second moment.}
%which are consistent  with the second moment. 
Through the lense of representation theory, we show in Theorem~\ref{prop.ambiguity}
that the second moment determines the signal up to a set of unitary matrices,  whose dimension is governed by the
decomposition of the space of signals into irreducible representations
of the group. While the unitary matrix ambiguities have been identified before in some special cases~\cite{kam1980reconstruction,bhamre2015orthogonal}, we show that the same pattern of ambiguities governs all MRA models.
Section~\ref{sec:ambiguities_examples} provides specific examples.

To resolve these ambiguities, we suggest assuming the signal is sparse under some basis. This is a common assumption in many problems in signal processing and machine learning,   such as regression~\cite{tibshirani1996regression,goodfellow2016deep}, compressed sensing~\cite{donoho2006compressed,candes2006robust,eldar2012compressed}, and various image processing applications~\cite{elad2010sparse}.
\rev{Note that the representations of compact groups that we consider are typically spaces of $L^2$ functions on a domain such as $\R^3$. As such, they do not come equipped with a canonical basis, so the assumption we make is that our signal is sparse with respect to a generic basis. The notion of generic basis comes from algebraic geometry and makes use of the fact that the set of
  all possible bases of a vector space is an algebraic variety.
%\tb{I am not sure I understand the second part of this sentence}.
When we say that a result holds for a generic basis, it means that there is 
a Zariski open set of bases for which the statement of the result holds. In particular, it holds for almost all bases. For more detail, see Section \ref{sec:sparsity_conditions}.
}

Our second contribution, presented  in Section~\ref{sec:sparsity} and summarized in Theorem~\ref{thm.mainMRA}, 
describes   the sparsity level under which the orbit of a 
generic sparse signal can be recovered from the second moment. That is, the sparsity level that allows resolving the unknown unitary matrices. 
This implies that merely $n=\omega(\sigma^4)$ observations are required for accurate signal recovery.
The sparsity level is bounded by a factor that depends on the dimensions of the irreducible representations and their multiplicities.
The proof of Theorem~\ref{thm.mainMRA} relies on tools from algebraic geometry and representation theory.
Specific results are provided in Section~\ref{sec:examples_sparsity}.

\paragraph{Implications to cryo-EM}
In Section~\ref{sec:cryoEM},  
we show that the second moment of the cryo-EM model~\eqref{eq:cryoEM_model} is the same as of the MRA model~\eqref{eq:mra}, when $G$ is the group of three-dimensional rotations SO(3) and $V$ is the space of band-limited functions on the ball.
Namely, the tomographic projection operator~\eqref{eq:tomographic_projection} does not change the second moment of the observations.   
We introduce this model in detail in Section~\ref{sec:cryoEM} and particularize the main result of this paper to cryo-EM in Theorem~\ref{thm.mainCRYO}. We now state this result informally.

\begin{theorem}[Informal theorem for cryo-EM] \label{thm:cryoEM_informal}
	In the cryo-EM model~\eqref{eq:cryoEM_model} (described in detail in Section~\ref{sec:cryoEM_model}), a generic $K$-sparse function $f \in V$ 
	is uniquely determined by the second moment for $K\lessapprox
	N/3$, where  $N = \dim V$.
\end{theorem}

Theorem~\ref{thm:cryoEM_informal} implies that sparse structures can 
be recovered, in the high noise regime, with only $n=\omega(\sigma^4)$ observations, improving upon $n=\omega(\sigma^6)$ for generic structures~\cite{bandeira2017estimation}. 
Figure~\ref{fig:wavelet} shows the distribution of wavelet coefficients (a standard choice of basis in many signal processing applications~\cite{mallat1999wavelet}) of a few molecular structures. 
Evidently, less than 1/3 of the coefficients capture almost all the
energy of the volumes, suggesting that the bound of Theorem~\ref{thm:cryoEM_informal} is  reasonable for typical molecular structures. 

A recent paper~\cite{bendory2022autocorrelation} showed 
that a structure composed of ideal point masses (possibly convolved with a kernel with a non-vanishing Fourier transform) can be recovered from the second moment. However, the technique of~\cite{bendory2022autocorrelation} is tailored for this specific model.
The same paper also suggests to recover a 3-D structure from the second moment based on a sparse expansion in a wavelet basis.
\rev{Our result implies that for a given wavelet basis then with probability one the generic signal whose expansion is sufficiently sparse with respect to that basis can be recovered from its second moment. Moreover, for a given basis, there is, in principle, a computational technique to test whether Theorem~\ref{thm:cryoEM_informal} holds for that basis. See Remarks~\ref{rem.generic}
  and~\ref{rem.test} for more detail.
 }
%Importantly, since our results hold for a generic choice of basis, it is not necessarily true for any specific choice, such as a wavelet basis.

\begin{figure}[h!]
	\centering
	\begin{subfigure}[h]{0.45\textwidth}
		\centering
		\includegraphics[scale=0.4]{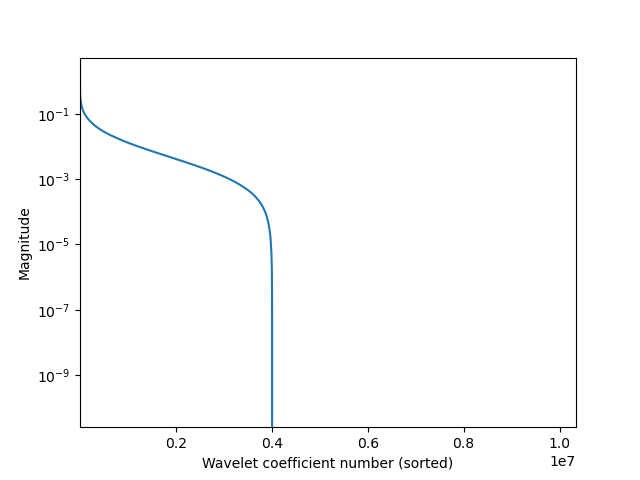}
		\caption{\texttt{EMD-2660}}
	\end{subfigure} \hfill 
	\begin{subfigure}[h]{0.45\textwidth}
		\centering
		\includegraphics[scale=0.4]{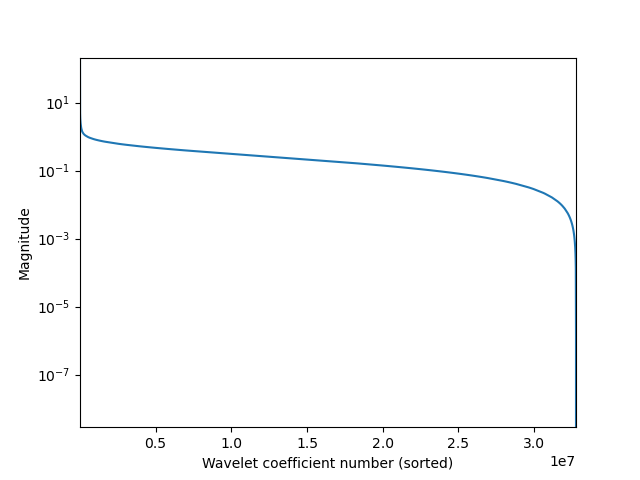}
		\caption{\texttt{EMD-8012}}
	\end{subfigure} \hfill %\hspace{2pt}
	
	\begin{subfigure}[h]{0.45\textwidth}
		\centering
		\includegraphics[scale=0.4]{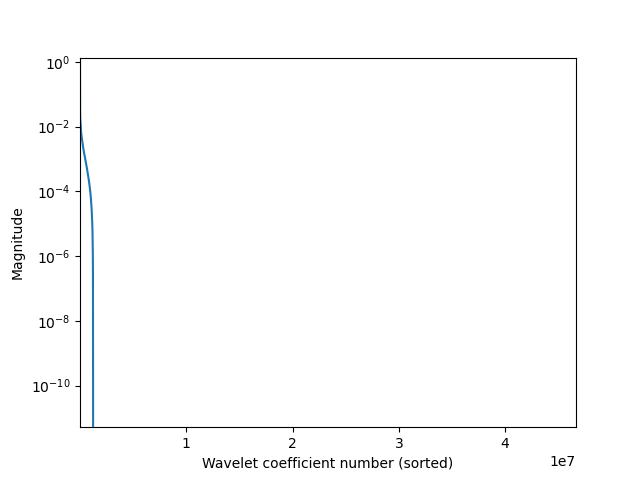}
		\caption{\texttt{EMD-8511}}
	\end{subfigure} \hfill
	\begin{subfigure}[h]{0.45\textwidth}
		\centering
		\includegraphics[scale=0.4]{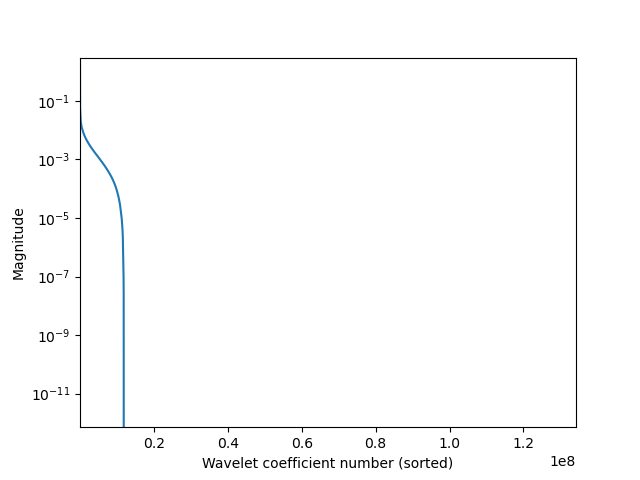}
		\caption{\texttt{EMD-2984}}
	\end{subfigure}
	\caption{\label{fig:wavelet} 
		The sorted wavelet coefficients of  cryo-EM structures whose experimental images are presented in Figure~\ref{fig:cryo}. The structures were downloaded from the   Electron Microscopy Data Bank (EMDB) \texttt{https://www.ebi.ac.uk/emdb}.  The structures were expanded using Haar wavelets, where the number of coefficients is approximately the same as the number of voxels. 
		Besides \texttt{EMD-8012}, all the volumes energy (i.e., the squared norm of the coefficients) is captured by less that one third of the coefficients, which is the bound of Theorem~\ref{thm:cryoEM_informal}. For \texttt{EMD-8012}, the same fraction of wavelet coefficients captures  more than 91\% of its energy.}
\end{figure}

\paragraph{Crystallographic phase retrieval}
The second moment of the  MRA model, where random elements of the group of circular shifts $\Z_N$ act on real signals in $\R^N$, is equivalent to the \rev{squared} absolute values of the Fourier transform of the signal,  known as the  power spectrum. 
Recovering a signal from its power spectrum is called the phase retrieval problem and it has numerous applications in signal processing; see recent surveys and references therein~\cite{shechtman2015phase,bendory2017fourier,grohs2020phase,bendory2022algebraic}. 

\rev{Crystals are often modeled as functions on a finite abelian group (typically $\Z_N$), which corresponds to the regular representation of
  the group. For this representation there is a natural notion of sparsity which
  corresponds to requiring that the function is non-zero only on a small subset
  of elements of the group. Real valued functions on $\Z_N$ are identified with $\R^N$ and this notion of sparsity corresponds to sparsity in the standard basis of $\R^N$.}
Recovering a sparse signal from the power spectrum is the main computational challenge in X-ray crystallography: a leading method for elucidating the atomic structure of molecules. 
This is by far the most important phase retrieval application. 
We discuss this problem in detail in Section~\ref{sec:phase_retrieval} 
\rev{and explain how the techniques of this paper can be used to prove
  that a $K$-sparse signal $f \in \R^N$, under a generic basis, 
  can be recovered from its power spectrum provided that $K \leq N/2$.}
%Our main result regarding the crystallographic phase retrieval problem states that  a  $K$-sparse signal $f\in\rev{\C}^N$  with known support under a
%generic basis can be recovered from its power spectrum, up to a global phase, as long as $K \leq N/2$.\dan{A lot of what is in Section~\ref{sec:phase_retrieval}
%  is now superced by the results of \cite{edidin2023generic}. Perhaps we can
%  shorten this section, since space is at a premium.}

\paragraph{Organization of the paper}
The rest of the paper is organized as follows. 
Section~\ref{sec:second_moment} formulates the second moment of the MRA model~\eqref{eq:mra} and shows that it determines the signal up to a set of unitary matrices.  The section also provides several examples. 
Section~\ref{sec:sparsity} derives a bound on the sparsity level that allows
\rev{for} unique recovery from the second moment (Theorem~\ref{thm.mainMRA}) in terms of the dimension and multiplicity of the irreducible representations, 
and provides examples. 
Section~\ref{sec:cryoEM} focuses on cryo-EM: the main motivation of this paper. 
We formulate the cryo-EM model in detail,  derive explicitly the ambiguities of the second moment, and deduce sparsity conditions allowing unique recovery (Theorem~\ref{thm.mainCRYO}). Section~\ref{sec:phase_retrieval} discusses the  crystallographic phase retrieval problem.
Section~\ref{sec:discussion} concludes this work, and delineates future research directions. %Appendix~\ref{sec:representation_theory} provides necessary background in representation theory.
The supplementary material provides necessary background in representation theory.

\section{The second moment and symmetries} \label{sec:second_moment}
This section lays out the mathematical background for the second moments of the MRA model~\eqref{eq:mra} for a compact group $G$ acting on an $N$-dimensional
real or complex
vector space~$V$.
Following standard terminology, we refer to a vector space $V$ equipped with an action of a group $G$ as a {\em representation} of $G$.
Our goal is to use \rev{classical} methods from the representation theory of compact groups to understand the information obtained from the second moment.
%In Appendix~\ref{sec:representation_theory}, we provide a necessary background in representation theory.
In the supplementary material, we provide a necessary background in representation theory.

Any representation of a compact group is {\em unitary}.
This means that elements of $G$ act on $V$ as unitary transformations.
\rev{In particular, the action preserves a Hermitian inner product.
By Weyl's unitarian trick, this inner product can be obtained
 by averaging any chosen inner product on $V$ over the group.}
If $V$ is a real vector space, then the action \rev{of} $G$ is {\em orthogonal},  meaning that
elements of $G$ act by orthogonal transformations.

Let $\K=\R$ or $\K=\C$ be the field.
Assuming that the distribution on the group $G$ is uniform (Haar), then
a choice of basis for an $N$-dimensional representation
$V=\K^N$ expresses the second moment as a function $\K^N \to \K^{N^2}$
given by the formula
\begin{equation} \label{eq.2ndmomentcoord}
	f \mapsto \int_G (g\cdot f) (g \cdot f)^* dg,
\end{equation}
where $\K^{N^2}=\Hom(V,V)$ is the vector space of linear transformations $V \to V$. 
The vector $f$ is viewed as column vector so
$(g \cdot f) (g \cdot f)^*$ is an $N \times N$, rank-one Hermitian matrix.

The second moment can also be defined without the use of coordinates, using tensor notation, as
a map $V \to V \otimes V^*$,
\begin{equation} \label{eq.2ndmomentabs}
	f \mapsto \int_G (g\cdot f) \otimes \overline{(g \cdot f)}dg.
\end{equation}
We will use both~\eqref{eq.2ndmomentcoord} and~\eqref{eq.2ndmomentabs} interchangeably. The reason that these formulations are equivalent is that
there is isomorphism of representations $V \otimes V^* \to \Hom(V,V)$ as discussed in the supplementary material. 
%in Appendix \ref{sec.adjoint}.
 If we choose an orthonormal basis
for $V$, then the tensor $f_1 \otimes \overline{f_2}$ corresponds to the matrix
$f_1f_2^*$.

Ultimately, we will view elements of $V$ as functions $D \to \C$, where
$D$ is some domain on which $G$ acts. 
For example, in cryo-EM $G= SO(3)$,
and $V$ is the subspace of $L^2(\C^3)$
consisting of the
Fourier transforms of real-valued functions
in $L^2(\R^3)$; this problem is discussed in detail in Section~\ref{sec:cryoEM}. 
The second moment of a function $f: D \to \C$ can be viewed as the
function $m^2_f \colon D\times D \to \C$, where 
\begin{equation} \label{eq.2ndmomentfunction}
	m^2_f(x_1,x_2) = \int_G (g \cdot f(x_1))\overline{(g \cdot f(x_2))} \;dg,
\end{equation}
where $g\cdot f  \colon D \to \rev{\C}$
is defined by $gf(x) = f(g^{-1}x)$.

\subsection{The second moment of an irreducible representation of $G$}

Recall that a representation is {\em irreducible} if it has no non-zero proper
$G$-invariant subspaces. \rev{Examples of reducible and irreducible representations are given in the supplementary material.}
If the representation $V$ is irreducible, then the following
proposition shows
that the second moment gives very little information about a vector $f \in V$.
\begin{proposition} 
	Let $V$ be an $N$-dimensional irreducible unitary representation of a compact group $G$ and identify $V$ with $\C^N$ via a choice of orthonormal basis
	$f_1, \ldots , f_N$ of $V$.
	Then, as a map $\C^N \to \C^{N^2}$, the
	second moment is given by the formula
	\begin{equation} \label{eq.matrixform}
		f \mapsto \rev{{|f|^2\over{N}}} I_N,
	\end{equation}
	where $I_N$ is the $N \times N$ identity matrix.
	In tensor notation, the second moment is the
	map $V \mapsto V \otimes V^*$ given by 
	\begin{equation} \label{eq:second_moment}
		f \mapsto  \rev{{|f|^2\over N}}\sum_{i=1}^N f_i \otimes \overline{f_i}.
	\end{equation}
	
\end{proposition}
\begin{proof}
	If we identify the Hermitian matrix $m^2_f = \int_G (g \cdot \rev{f}) (g \cdot \rev{f})^*dg $ as giving a linear transformation
	$V \to V$,  then the second moment defines a map $V \to \Hom(V,V)$,
	where $\Hom(V,V)$ is the group of linear transformations $V \to V$. 
	Since
	the second moment is by definition invariant under the action of $G$
	on $V$ (i.e., $f$ and $g\cdot f$ both yield the matrix $m^2_f$), 
	the matrix $m_f^2$
	defines a $G$-invariant linear transformation on $V$.
	However since $V$ is irreducible,
	by Schur's Lemma, any $G$-invariant
	linear transformation $V \to V$ is a scalar multiple of the identity.
        \rev{Since $G$ acts by unitary transformations, $\trace ((g\cdot f) (g\cdot f)^*) =
          \trace (f f^*) = |f|^2$ for any $g \in G$. Thus,
          $$\trace m^2_f = \int_G \trace ((g\cdot f) (g\cdot f)^*)\;dg =|f|^2.$$}

	The formula~\eqref{eq:second_moment}
	is equivalent to the first formula because under the identification of
	$V \otimes V^*$ with $\Hom(V,V) = \K^{N^2}$, 
	the tensor $\sum_{i=1}^N f_i \otimes \overline{f_i}$ corresponds to the
	identity matrix.
\end{proof}

\subsection{The second moment for multiple copies of an irreducible representation}  \label{sec:multiple_copies}

The following discussion is motivated by the situation in cryo-EM, where we view
$\R^3$ as a collection of spherical shells.
%Thus, SO(3)
%acting on $L^2(\R^3)$ implies taking a number of copies of $L^2(S^2)$.
\rev{In other words, we model $SO(3)$ acting on $L^2(\R^3)$ by taking a
  number of copies of $L^2(S^2)$.}
This is a standard model in cryo-EM, and is  introduced in detail in Section~\ref{sec:cryoEM}. 

Consider the case where the representation
$V$ decomposes as the direct
sum of $R$ copies of a single irreducible representation $V_0$.
In other words, there is a $G$-invariant isomorphism 
$V \simeq V_0^{\oplus R}$. This means that 
any vector 
$f \in V$
can be decomposed uniquely
as
$f = f[1] + \ldots + f[R]$, with $f[r]$
in the $r$-th copy of $V_0$.
The summands are invariant under the action of $G$ so $(g\cdot f)[r] =
g\cdot f[r]$.

Since $V$ decomposes as the sum $V_0^{\oplus R}$, the tensor product
$V \otimes V^*$ decomposes as the sum of tensor products
$\oplus_{i,j=1}^RV_0[i] \otimes V_0^*[j]$, where $V_0[r]$ indicates the $r$-th
copy of $V_0$ in the decomposition of $V$.
In particular, using tensor notation for the second moment, 
we can decompose
\begin{eqnarray} \label{eq:second_moment_multiple}
	m^2_f  =  \int_G (g\cdot f) \otimes \overline{(g\cdot f)} dg= &\sum_{i,j=1}^R m^2_f[i,j], 
\end{eqnarray}
where
\begin{equation} \label{eq.m2decomp}
	m^2_f[i,j]  =  \int_G (g\cdot f[i]) \otimes \overline{(g\cdot f[j])}dg \in
	V_0[i] \otimes V^*_0[j],
\end{equation}
is the component in the $(i,j)$-th summand of the tensor product $V \otimes V^*$.
Each of the summands in \eqref{eq.m2decomp} defines a $G$-invariant linear
transformation 
$V_0[i] \to V_0[j]$.

Let  $N_0 = \dim V_0$.
For suitable orthonormal bases $f_1[i], \ldots , f_{N_0}[i]$
and $f_1[j], \ldots , f_{N_0}[j]$ of $V_0[i]$ and $V_0[j]$, respectively, Schur's Lemma implies that 
$$m^2_f[i,j]=\int_G (g\cdot f[i]) \otimes \overline{(g \cdot f[j])} \;dg =
\rev{{\langle f[i], f[j] \rangle\over{N}}} \left( \sum_{k=1}^{N_0} f_k[i] \otimes \overline{f_k[j]}
\right).$$
To put this more directly, if we view an element of $V = V_0^{\oplus R}$
as an $R$-tuple
$f[1], \ldots , f[R]
$ of elements of $V_0$, then the second moment determines
all pairwise inner products $\langle f[i],  f[j] \rangle$.
Equivalently, if we consider the vectors $f[1], \ldots , f[R]$
as the column vectors of an $N_0  \times R$ matrix~$A$, then the second moment
determines the $R \times R$ Hermitian matrix $A^*A$. 
\rev{	Therefore, the vectors $f[1], \ldots , f[R] \in V_0$ are determined
	from their pairwise inner products up to the action of the unitary group
	$U(N_0)$, parameterizing the isometries of~$V_0$.
	If, as will be the case for cryo-EM, we know that each $f[r]$ lies in a conjugation invariant subspace
	of $V$ (for example, it is the Fourier transform
	of a real vector), then we can determine each $f[r]$ up to the action
	of a subgroup of $U(N_0)$, isomorphic to the real orthogonal group~$O(N_0)$.}
	
\subsection{The second moment of a general finite dimensional representation and its group of ambiguities}

A general finite dimensional representation of a compact group
can be decomposed as 
\begin{equation}
	V = \oplus_{\ell = 1}^L V_\ell^{\oplus R_\ell},
\end{equation}
with the $V_\ell$ are
distinct (non-isomorphic)
irreducible representations of $G$ of dimension $N_\ell$. 
An element of $f \in V$ has a unique $G$-invariant decomposition as a sum
\begin{equation} 
	f = \sum_{\ell = 1}^L \sum_{i= 1}^{R_\ell} f_\ell[i], \label{eq.decomp}
\end{equation} 
where $f_\ell[i]$ 
is in the $i$-th copy of the irreducible representation $V_\ell$.
In this case, the second moment decomposes as a sum of tensors
$\displaystyle{\int_G (g \cdot f_\ell[i]) \otimes \overline{(g \cdot f_{m}[j])}dg}$.
Each of these tensors
determines a $G$-invariant map $V_\ell[i] \to V_m[j]$. Since  $V_\ell[i]$ and $V_m[j]$ are non-isomorphic irreducible
representations, Schur's Lemma implies that there are no non-zero
$G$-invariant linear transformations
$ V_\ell[i] \to V_m[j]$
for $\ell\neq m$. In other words, 
we have a generalized orthogonality relation that the tensors
$\displaystyle{\int_G (g \cdot f_\ell[i]) \otimes\overline{ (g \cdot f_m[j])} \;dg}$ 
are zero if $\ell \neq m$
for all $i,j$.

Hence, the second moment decomposes as a sum
\begin{equation} \label{eq.general_second_moment}
	m_f^2 = \sum_{\ell =1}^L \sum_{i,j=1}^{R_\ell}
	\rev{{\langle f_\ell[i], f_\ell[j] \rangle\over{N_\ell}}}
	\left( \sum_{k=1}^{N_\ell} f_{k,\ell}[i] \otimes \overline{f_{k,\ell}[j]} \right),
\end{equation}
where the vectors  $f_{1, \ell}[i], \ldots , f_{N_\ell, \ell}[i]$
form
an orthonormal basis for the $i$-th copy of the $\ell$-th irreducible
representation $V_\ell$.
\rev{
\begin{remark} \label{rem.endomorphism}
 The second moment is a map
 $V \to \Hom_G(V,V)$. As noted by a referee, $\Hom_G(V,V)$ is the endomorphism ring of the $G$-module $V$. A result from classical representation theory, which follows from Schur's Lemma, 
states that this ring decomposes into a sum of matrix algebras
$\oplus_{\ell =1}^L \Mat(R_\ell)$ and our description of
the second moment can also be derived using this decomposition.
\end{remark}
}

\subsubsection{Functional representation of the second moment}
If, as will be the case for our model of cryo-EM, we view the
elements of $V$ as functions $f \colon D \to \C$, then we can reformulate
\eqref{eq.general_second_moment} as follows.
Suppose that $f_{1,\ell}[i], \ldots , f_{N_l,\ell}[i]$ are functions
$D \to \C$ which form an orthonormal basis for the $i$-th copy of the
$\ell$-th irreducible representation $V_\ell$. If we expand $f_{\ell}[i] =
\sum_{m=1}^{N_\ell} A^m_\ell[i] f_{m, \ell}[i]$, then the second moment
realized as a function $D \times D \to \C$ is expanded as
\begin{equation} \label{eq.functional_second_moment}
	m_f^2(x_1,x_2) = \sum_{\ell =1}^L \sum_{i,j=1}^{R_\ell}
	\left( \sum_{m = 1}^{N_\ell} A^m_\ell[i] \overline{ A^m_\ell}[j]\right)
	\left( \sum_{k=1}^{N_\ell} f_{k,\ell}[i](x_1) {f_{k,\ell}^*[j]}(x_2) \right),
\end{equation}
where $x_1, x_2$ are,  respectively, the variables on the first and second copies of $D$ respectively.

\subsubsection{The group of ambiguities} \label{sec:group_ambiguities}

The main result of this section is a characterization of the group of ambiguities
of the second moment. Later on, we provide a few explicit examples.

\rev{Suppose that $V$ decomposes as a sum of irreducible representations $V = \oplus_{\ell=1}^L V_\ell^{R_\ell}$,
  where $\dim V_\ell= N_\ell$, and let $H = \prod_{\ell = 1}^L{U(N_\ell)}$.
  The group $H$ acts on $V$ as follows. If $f \in V$ is represented by
  an $L$-tuple of $(A_1, \ldots , A_L)$ with $A_\ell$ an $N_\ell \times R_\ell$ matrix
  and $h = (U_1, \ldots , U_L)$ with $U_\ell\in U(N_\ell)$, then
  $h\cdot f = (U_1A_1, \ldots , U_L A_L)$.

\begin{theorem} \label{prop.ambiguity}
With the notation as above, a vector $f\in V$ is determined from the second moment up to the action
of the ambiguity group $H = \prod_{\ell = 1}^L{U(N_\ell)}$. That is,
$m^2_f = m^2_{f'}$ if and only $f = h \cdot f'$ for some $h \in H$.
\end{theorem}
}
\begin{proof}
	If we decompose a vector $f \in V$  
	as in~\eqref{eq.decomp},
	then the second moment \eqref{eq.general_second_moment} determines the inner
	products  
	$\langle f_\ell[i],f_\ell[j]\rangle $ 
	for all $\ell = 1, \ldots , L$ and $i,j \in {1,\ldots R_\ell}$.
	
	For a general representation, an element of $V$ can be represented by an $L$-tuple
	$(A_1, \ldots , A_L)$, where $A_\ell$ is an 
	$N_\ell \times R_\ell$ complex matrix corresponding to an element in the summand $V_\ell^{\oplus R_\ell}$. 
	The second moment determines the $L$-tuple of $R_\ell \times R_\ell$
	Hermitian matrices $(A_1^* A_1, \ldots , A_L^* A_L)$.
\rev{        Thus, if $U_1, \ldots , U_L$ are unitary matrices, then an $L$-tuple
        of $N_\ell \times R_\ell$-matrices
        $(U_1A_1, \ldots , U_LA_L)$ has the same second moment
        because $(U_\ell A_\ell)^*(U_\ell A_\ell)= A_\ell^* A_\ell$ for each
        $\ell$.} In particular, a vector $f$ is determined from the second moment up to the action
	of the product of unitary groups $\prod_{\ell =1}^L U(N_\ell)$. 
\end{proof}

\rev{
\begin{remark}[Parameter counting]
	Since each unitary matrix is determined by $N_\ell^2$ real parameters, the ambiguity group is of dimension $N_H=\sum_{\ell =1}^L N_\ell^2$. 
	If the ambiguity group is   isomorphic to the real orthogonal groups, as in cryo-EM, then the ambiguity group is of dimension $N_H=\sum_{\ell =1}^L N_\ell(N_\ell-1)/2$.
\end{remark}

\begin{remark} \label{rem:Rl} Note that the
	total dimension of the ambiguity group of the second moment does not 
	depend on the multiplicities $R_\ell$. In particular, the ratio of the  dimensions is  
	\begin{equation*}
		\frac{N_H}{N} = \frac{\sum_{\ell = 1}^L N_\ell^2}{2\sum_{\ell =1}^L R_\ell N_\ell}.
	\end{equation*} 
	This implies that  as the number of  multiplicities increases, the proportional amount of information
	about the signal contained in the second moment increases as well.
\end{remark}
}

\subsection{Examples} \label{sec:ambiguities_examples}
\subsubsection{The power spectrum} \label{ex.powerspectrum}
Consider the group $G = \Z_N$ acting on $\K^N$ by cyclic shifts,  where
$\K= \R$ or $\K = \C$. 
In the Fourier
domain, the cyclic group $G = \Z_N$ acts by multiplication by roots
of unity. In particular, we identify $\Z_N = \mu_N$, where $\mu_N$ is the $N$-th roots of unity. If $\omega \in \mu_N$, then 
\begin{equation}\label{eq:fourier_action}
	\omega \cdot (f[0], \ldots , f[N-1]) = (f[0], \omega f[1], \ldots , \omega^{N-1} f[N-1]).
\end{equation}
The vector space $\C^N$ with this action of $\mu_N$ decomposes
as a sum of one-dimensional irreducible representations (namely, $N_\ell=\rev{R_\ell}=1$ for all $k$ so that $N=L$) $V_0 \oplus \ldots \oplus V_{N-1}$, where $\omega \in \mu_N$ acts on $V_i$ by $\omega \cdot  f[i] =
\omega^i f[i]$.
The second moment of a vector $f \in \C^N$ in the Fourier domain is the power
spectrum $(f|[0]|^2, \ldots , |f[N-1]|^2)$. This determines
the vector up to the action of the group $(S^1)^{N}$ since $U(1) = S^1$.
\rev{Figure~\ref{fig:pr} shows an example of two different images with the same power spectrum.}

\begin{figure}[h!]
	\centering
	\begin{subfigure}[h]{0.45\textwidth}
		\centering
		\includegraphics[scale=0.4]{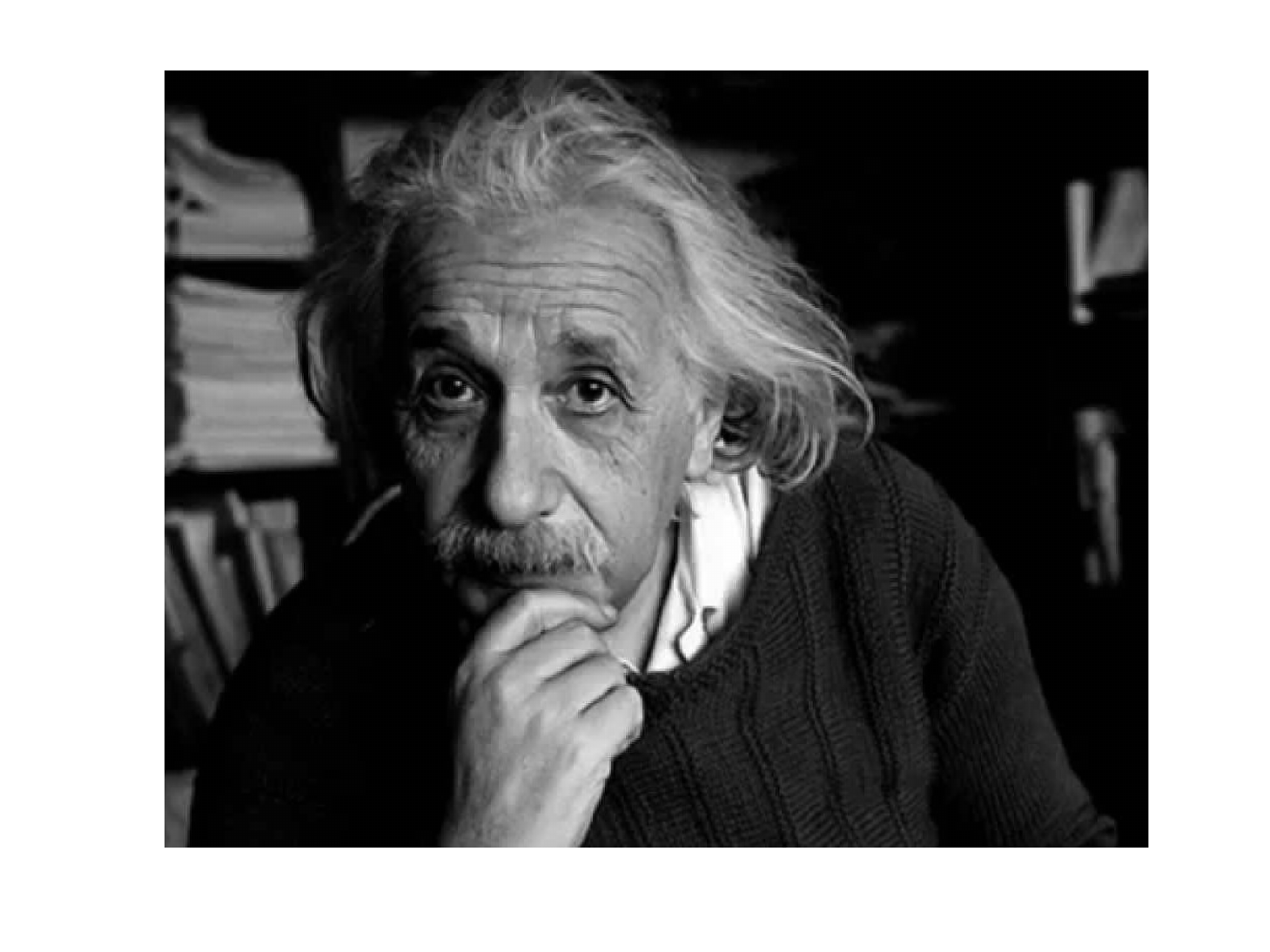}
		\caption{\texttt{Einstein}}
	\end{subfigure} \hfill 
	\begin{subfigure}[h]{0.45\textwidth}
		\centering
		\includegraphics[scale=0.4]{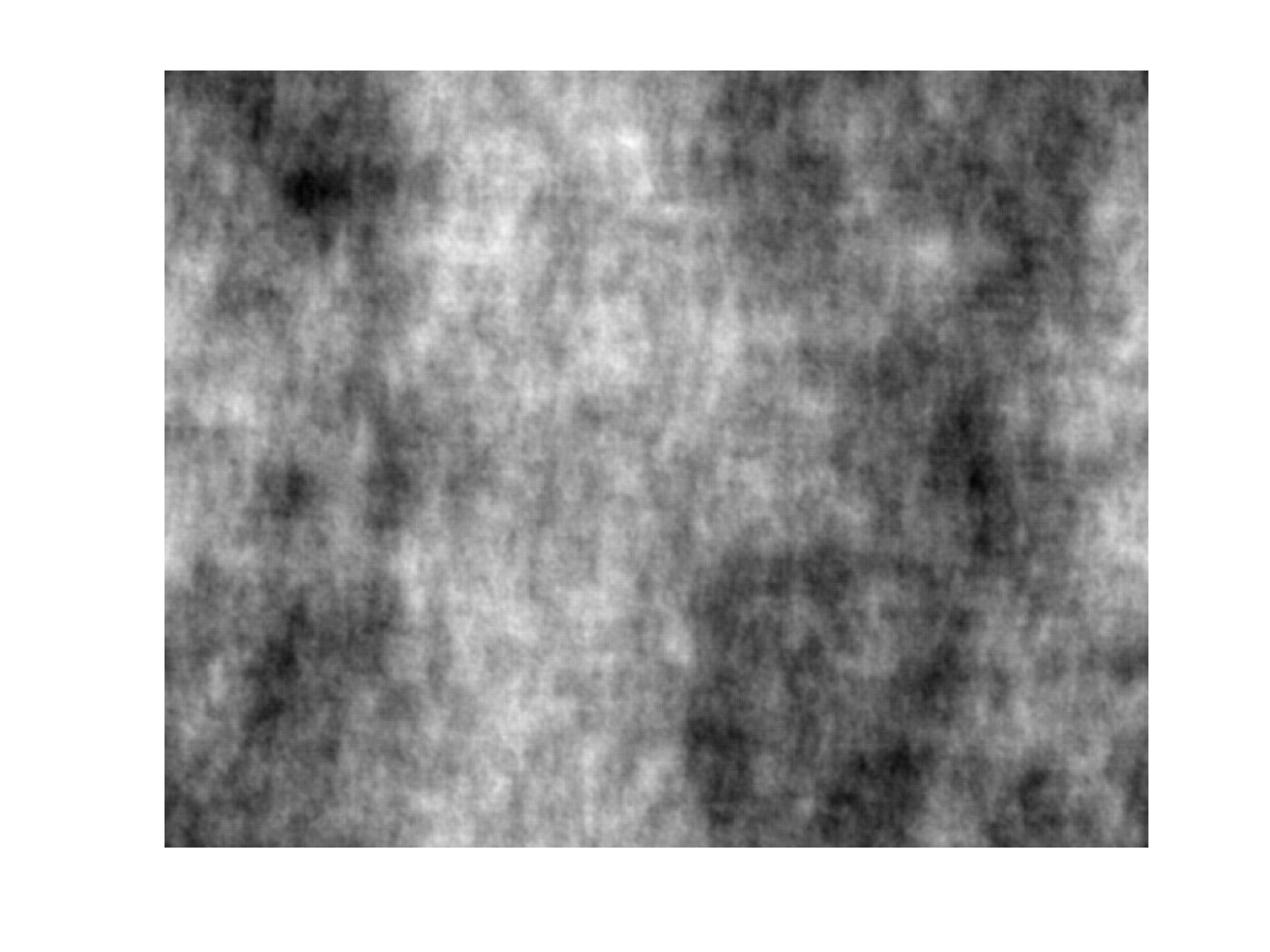}
		\caption{\texttt{Einstein with random phases}}
	\end{subfigure} 
	\caption{\label{fig:pr} \rev{The left panel shows an image of Albert Einstein. To generate the image of the right panel, we combined the absolute values of the Fourier transform of Einstein's image with random phases  and computed the inverse Fourier transform. This example underscores that two images with the same power spectrum may be very different. More generally, two signals which are equal up to a set of unitary matrices (e.g., have the same second moment) may differ significantly.}
	} 
\end{figure}

Recall that the image
of $\R^N$ under the discrete Fourier transform is the real subspace
of~$\C^N$ given by the condition  $f[N-i] = \overline{f[i]}$.
Thus, if $f$ is the Fourier transform of a real vector,  the ambiguity group
must preserves the condition that $f[N-i] = \overline{f[i]}$ and is therefore the subgroup of
\begin{equation} \label{eq:fourier_real}
	\{(\lambda_0, \ldots, \lambda_{N-1} | \lambda_{N-n} = \lambda_{n}^{-1}\} \subset (S^1)^N.
\end{equation}

Recovering a signal from its power spectrum is called the phase retrieval problem~\cite{shechtman2015phase,bendory2017fourier,grohs2020phase,bendory2022algebraic}; see Section~\ref{sec:phase_retrieval} for further discussion.

\subsubsection{Dihedral MRA}
Consider the action of the dihedral group $G= D_{2N}$ acting $\K^N$, where
the rotation $r \in D_{2N}$ acts by cyclic shift and the reflection
$s \in D_{2N}$ acts by $(s \cdot f)[i] = f[N-i]$.
In the Fourier
domain, $(s \cdot f)[i] = \rev{f[N-i]}$ and  
$(r \cdot f)[i] = \omega^i f[i]$  as in~\eqref{eq:fourier_action}.
%   In the Fourier
%   domain,  these actions are given by
%   $(r \cdot f)[i] = \omega^i f[i]$ and $(s \cdot f)[i] = \overline{f[N-i]}$.
%   Here, $\omega = e^{\iota 2\pi/N}$ is a primitive $N$-th root of unity. 
In~\cite{bendory2022dihedral}, it was shown that the orbit of a generic signal is determined uniquely from the second moment if the group elements are drawn from  a non-uniform distribution over the dihedral group. 
Here, we consider a uniform (Haar) distribution of the group elements.

The vector space $\C^N$ with this action of $D_{2N}$ decomposes into a sum of one and two-dimensional irreducible representations,  depending on the parity of $N$ \rev{(with multiplicity $R_\ell=1$)}. If $N$ is even, then 
$$\C^N = V_0 \oplus V_{1} \oplus \ldots V_{N/2-1} \oplus V_{N/2},$$
where $V_0$ is the one dimensional subspace spanned by the vector
$e_0 = (1,0 \ldots , 0)$, $V_{N/2}$ is spanned by the vector $e_{N/2}$ \rev{($N_0=N_{N/2}=1$)}, 
and for $1 \leq \rev{\ell}\leq N/2 -1$, $V_{\rev{\ell}}$ is the subspace spanned by $\{e_{\rev{\ell}}, e_{\rev{N-\ell}}\}$ \rev{($N_\ell=2$)}.
Similarly, if $N$ is odd, then
$$ \C^N = V_0 \oplus V_1 \ldots \oplus V_{(N-1)/2},$$
where again $V_0$ is spanned by $e_0$ and for $\rev{\ell} \geq 1$
$V_{\rev{\ell}},$ is spanned by $\{e_{\rev{\ell}},
e_{N-\rev{\ell}}\}$.

Therefore, the second moment of a vector $f$  in the Fourier domain
determines the $N/2 +1$ real numbers
$$(|f[0]|^2, |f[1]|^2 + |f[N-1]|^2, \ldots , |f[N/2 -1]|^2 + |f[N/2 +1]|^2, |f[N/2]|^2)$$
if $N$ is even, and the \rev{$(N+1)/2$} real numbers
$$(|f[0]|^2, |f[1]|^2 + |f[N-1]|^2, \ldots , |f[(N-1)/2]^2 + |f[(N+1)/2]|^2)$$
if $N$ is odd.
When $\K = \C$, this is less information than the power spectrum.
When $N$ is even, the ambiguity group is
$S^1 \times U(2)^{N/2} \times S^1$ and when $N$ is odd the ambiguity group
is $S^1 \times U(2)^{(N-1)/2}$. 
However,
if $\K = \R$ then the second moment gives the power spectrum because
if $f$ is the Fourier transform of a real vector then we have $|f[i]| = |f[N-i]|$.
In this case, the ambiguity group is $\pm{1} \times O(2)^{N/2} \times \pm{1}$
if $N$ is even and if $N$ is odd then it is $\pm{1} \times O(2)^{(N-1)/2}$.
These groups are isomorphic to the subgroups of $(S^1)^N$ considered in~\eqref{eq:fourier_real}.

\subsubsection{MRA with rotated images} \label{ex.rotatedimage}
In this model the Fourier transform of an image  is represented as a radially discretized band limited function on $\C^2$.
That is, our function $f$ is expressed as
$f = (f[1], \ldots , f[R])$, where 
\begin{equation} \label{eq.fouriershell}
	f[r](\theta) = \sum_{k=-L'}^{L'} a_{k,r}e^{\I\theta k}, \quad \theta\in[0,2\pi),
\end{equation}
for some bandlimit $\rev{L'=(L-1)/2}$ and $R$ radial samples. 
The action of a rotation $S^1$ on the image is given by 
\begin{align*}
	e^{\iota \alpha} \cdot f[r](\theta)   =  \sum_{k=-L'}^{L'} a_{k,r} e^{\I(\theta -\alpha) k} 
	=  \sum_{k=-L'}^{L'} a_{k,r} e^{-\I \alpha k} e^{\I \theta k}.
\end{align*}
With this action, the parameter space of two-dimensional images 
is the $S^1$-representation \rev{$V = V^{\oplus R}_{-L'} \oplus \ldots \oplus V^{\oplus R}_{L'}$},  where
$V_{k}$ is the one-dimensional representation of $S^1$, where $e^{\I\alpha} \in S^1$ acts
with weight $-k$. \rev{Namely, $N_k=1$, $R_k=R$ for all $k$ so that $N = LR$.}
The $(r_1,r_2)$ component of the second moment equals 
\begin{align} \label{eq:2dmra}
	m^2_f[r_1,r_2](\theta_1,\theta_2)  &=
	\int_\alpha e^{-\I \alpha} f[r_1](\theta_1)
	\overline{e^{-\I \alpha} f[r_2](\theta_2)} d\alpha \nonumber \\ & 
	\int_\alpha \sum_{k_1=-L'}^{L'} a_{k_1,r_1} e^{\I(\theta_1-\alpha) k_1} \sum_{k_2=-L}^L \overline{a_{k_2,r_2}} e^{-\I(\theta_2-\alpha) k_2}d\alpha \nonumber  \\ & =  \sum_{k=-L'}^{L'}a_{k,r_1} \overline{a_{k,r_2} }e^{\I(\theta_1 - \theta_2) k} \\ & =  \sum_{k=-L'}^{L'}a_{k,r_1} \overline{a_{k,r_2}} e^{\I\Delta \theta k}\nonumber
	\\ & =   m^2_f[r_1,r_2](\Delta\theta),\nonumber
\end{align}
where $\Delta\theta:=\theta_1-\theta_2.$
Following our previous discussion, a function $f \in V$ is determined by
a \rev{$L$}-tuple of  $1\times R$ 
matrices $(A_{-L'}, \ldots, A_{L'})$, where
$A_k = (a_{k,1} , \ldots , a_{\rev{k},R})^T$. 
The second moment computes
the \rev{$L$}-tuple of rank-one $R \times R$ matrices $(A_{-L'}^* A_{-L'}, \ldots , A_{L'}^*A_{L}')$.
Since each irreducible summand in the representation $V$ has dimension one
(namely, $N_\ell=1$ for all $\ell$), the ambiguity group of the second moment for the rotated images problem is \rev{$(S^1)^{L}$}.
If we assume that
the function $f$ is the Fourier transform of a real valued function, then $a_{k,r} = \overline{a_{-k,r}}$
and the ambiguity group is \rev{$O(2)^{L'} \times \pm 1$}.

\subsubsection{Two-dimensional tomography from unknown random projections} \label{sec.2dtomography}

The problem of recovering a two-dimensional image from it tomographic projections is a classical  problem in computerized tomography (CT) imaging~\cite{natterer2001mathematics}.
However, in some cases, the viewing angles are unknown and may be considered random.
Due to the Fourier Slice Theorem, this is equivalent to randomly rotating the image, and then acquiring a single one-dimensional line of its Fourier transform. 
While generally an image cannot be recovered from such random projections (in contrast to the three-dimensional counterpart~\eqref{eq:cryoEM_model}, where recovery is theoretically possible based on the common-lines property~\cite{singer2011three,shkolnisky2012viewing}), it was shown that unique recovery, up to rotation, requires  rather mild conditions~\cite{basu2000uniqueness}. Different algorithms were later developed, see for example~\cite{coifman2008graph,singer2013two,zehni2022adversarial}.  

In this model, we compute the second moment of the Fourier transform of the  image after tomographic projection to
a line. In other words, we compute the integral
\begin{equation*}
	\int_{S^1} T e^{\I \alpha} \cdot f[r_1](\theta_1) T \overline{e^{\I \alpha} \cdot f[r_2] (\theta_2) }d\alpha,
\end{equation*}
where $T$ is the tomographic projection to the line $\theta= 0$ (the two-dimensional counterpart of~\eqref{eq:tomographic_projection}). 
\rev{Because we are computing the second moment after tomographic projection,
  we cannot directly determine the ambiguity group
  from Theorem~\ref{prop.ambiguity}. 
  In this case, the tomographic projection causes us to lose information
  and}
we obtain a function only of $r_1, r_2$ (compare with~\eqref{eq:2dmra})
\begin{equation*}
	m^2_f[r_1,r_2] = \sum_{k=-L'}^{L'} a_{k,r_1} \overline{a_{k,r_2}},
\end{equation*}
\rev{where ${L'=(L-1)/2}$.}
If we view the \rev{$L$}-tuple of $1\times R$ %$R\times 1$ 
matrices $(A_{-L'}, \ldots , A_{L'})$
as a single $\rev{L} \times R$-matrix~$A$, then the projected second
moment determines
the Hermitian matrix $A^* A$. Equivalently, an element of $V$ is determined
by $R$ vectors in \rev{$\C^{L}$} and the projected second moment
determines all pairwise inner products of these vectors. In this
case, \rev{the loss of information caused by the tomographic projection
 means that the ambiguity group is the bigger group} $\rev{U(L)}$
(or \rev{$O(L)$} if the image is the Fourier transform of a real-valued function) compared to $\rev{(S^1)^{L}}$ in the unprojected case (respectively,
\rev{$O(2)^{L'} \times \pm 1$}).

\rev{
\begin{remark} Note that when $G=\SO(3)$, the second moment
is unchanged by the tomographic projection from $\R^3 \to \R^2$.
See Lemma \ref{lem:2ndmoment_cryo-EM}.
\end{remark}
}  

\section{Retrieving the unitary matrix ambiguities for sparse signals} \label{sec:sparsity}

In the previous section, we have shown that it is generally impossible to recover a vector~$f$ in a representation $V$ of a
compact group $G$ from its second moment due to the large group of ambiguities.
To resolve these ambiguities and recover the  signal in either the MRA~\eqref{eq:mra} or cryo-EM~\eqref{eq:cryoEM_model} models, we need a prior on the sought signal.
In this work, we assume that the signal is sparse in some basis. This assumption has been studied and harnessed in the MRA~\cite{bendory2022sparse,ghosh2022sparse} and cryo-EM literature~\cite{bendory2022autocorrelation,vonesch2011fast,jonic2015coarse,kawabata2018gaussian,esteve2022spectral,zehni20203d}. 
In this section, we derive bounds on the sparsity level that allows retrieving the missing unitary matrices, as  a function of the
dimensions and multiplicities of the irreducible representations.
We also provide a couple of examples, and leave more detailed discussions on cryo-EM  and phase retrieval to, respectively, Section~\ref{sec:cryoEM} and Section~\ref{sec:phase_retrieval}.

\subsection{Sparsity conditions} \label{sec:sparsity_conditions}
Let $V$ be an $N$-dimensional
vector space. The notion of sparsity depends on the choice
of an orthonormal basis ${\mathcal V} = \{f_1, \ldots , f_N\}$. %${\mathcal V} = \{v_1, \ldots , v_N\}$. 
A vector $f \in V$ %$v \in V$ 
is $K$-sparse
with respect to this ordered basis if $f$ is a linear combination
of at most $K$ elements of this basis. The set of $K$-sparse vectors
with respect to an ordered basis ${\mathcal V}$ is the union of
$\binom{N}{K}$ linear subspaces $\L_S(\V)$, where $\L_S(\V)$ is the subspace spanned by
the vectors $\{f_i\}_{i \in S}$ 
and $S$ is a $K$-element subset of $[1,N]$.

Let
\begin{equation} \label{eq:V}
	V = \oplus_{\ell=1}^L V_\ell^{R_\ell},
\end{equation}
be a representation of a compact group
$G$, where $\dim V_\ell = N_\ell$ . Let $H = \prod_{\ell=1}^{L} U(N_\ell)$
be the ambiguity group
of the second moment (see Theorem~\ref{prop.ambiguity}).

The main result of this section is the following.

\begin{theorem} \label{thm.mainMRA}
	Let $V$ be a representation as in~\eqref{eq:V}, let $N= \sum_{\ell=1}^L N_\ell R_\ell$
	be its total dimension and let $M = \sum_{\ell=1}^L \min(N_\ell R_\ell, N_\ell^2)$. 
	Then, for a generic choice of orthonormal basis $\V$,  
	a generic $K$-sparse vector $f\in V$ with $K \leq N - M$
	is uniquely determined by its second moment, up to a global phase.
\end{theorem}

\rev{We note that, as in Remark~\ref{rem:Rl}, the larger the number of irreducible representation copies is, the easier the problem is. In particular, note that if $R_\ell\gg N_\ell$ for all $\ell$, the sparsity bound read $K\leq \sum_{\ell=1}^LR_\ell N_\ell=N$. That is, the sparsity level is proportional to the dimension of the representation. In Theorem~\ref{thm.mainCRYO} we provide an explicit example for the cryo-EM case.}

\rev{
  \begin{remark} \label{rem.generic}
    The set of ordered orthornormal bases of an $N$-dimensional
    vector space $V$ can be identified
    with the real algebraic group $O(N)$ if $V$ is real, and $U(N)$ if
    $N$ is complex. When we say that our result
    holds for a {\em generic basis} we mean that there is a real Zariski
    open subset of $O(N)$ (resp. $U(N)$) parametrizing bases for which the conclusion of Theorem \ref{thm.mainMRA} holds. Since the complement of a Zariski open set has Lesbegue measure zero, this means that given
    an orthonormal basis ${\mathcal V}$ for $V$, then with probability one
    Theorem~\ref{thm.mainMRA} holds for that basis.
  \end{remark}
  
  \begin{remark} \label{rem.test}
    Given a basis ${\mathcal V}= \{f_1, \ldots , f_N\}$ for $V$, we can
    express $f \in V$ as $\sum_{n=1}^N x_i f_i$
and the second moment is a collection of homogeneous
quadratic functions in $x_1, \ldots ,x_N$, which we denote
by $m^2_f(x_1, \ldots , x_N)$. The 
following computational test is a simple generalization
of the test used in \cite[Sections 4.3.3, 4.3.4]{bendory2020toward}
that can be used to decide whether ${\mathcal V}$ satisfies the theorem with a sparsity level of $K$:

If $S \subset [1,N]$ is a subset of size $K$, let
$$I_S = \{(x_1,\ldots , x_N), (y_1, \ldots , y_N)| m^2_f(x_1, \ldots , x_N) =
m^2_f(y_1, \ldots , y_N)\} \subset \L_{S}({\mathcal V}) \times \L_S({\mathcal V}),$$
where $\L_S(\V)$ is the subspace spanned by $\{f_i\}_{i \in S}$.
Likewise, if $S,S'$ are two distinct $K$-element subsets of $[1,N]$, let
$$I_{S,S'} = \{(x_1,\ldots , x_N), (y_1, \ldots , y_N)| m^2_f(x_1, \ldots , x_N) =
m^2_f(y_1, \ldots , y_N)\} \subset \L_{S}({\mathcal V}) \times
\L_{S'}({\mathcal V}).$$
The conclusion of Theorem~\ref{thm.mainMRA} holds if
$I_S$ has dimension exactly $K$ and degree one, and $\dim I_{S,S'} < K$.
For small values of $N$, these conditions can be checked using a computer algebra system, but not in polynomial time \cite[Appendix D]{bendory2020toward}.
\end{remark}
}  

\rev{
\begin{remark}[Frames]
	Recall that a collection ${\mathcal F}$ of vectors
	in a finite-dimensional vector space $V$ is a {\em frame} if the vectors
	span $V$.
	The methods used to prove Theorem~\ref{thm.mainMRA} can also be used to prove
	a corresponding result where orthonormal bases are replaced by arbitrary
	frames.
	The only difference between working with frames instead of bases
	is that definition of a vector being sparse with respect
	to an ordered frame is more subtle. The reason is that for a generic frame
	${\mathcal F}$, 
	any $N$-element subset consists of linearly independent vectors, so any
	$f \in V$ which has zero frame coefficients with respect to $N$ elements
	in ${\mathcal F}$ must necessarily be zero. In particular, if
	we work with frames, then the condition that a vector is $K$-sparse should
	be replaced by the condition that at least $N-K$ of the frame coefficients
	are zero. Otherwise, the statements and proofs remain the same.
\end{remark}
}

\subsubsection{Strategy and remarks on the proof}
The proof of Theorem~\ref{thm.mainMRA} involves a number of steps.
Suppose that ${\mathcal V}$ is a 
generic orthonormal basis and consider the set of vectors
which are $K$-sparse with respect to ${\mathcal V}$. The set of such vectors
form the union of $\binom{N}{K}$ $K$-dimensional linear subspaces
of $V$.
The strategy of the proof is to show that with bounds on $K$ given in
the Theorem~\ref{thm.mainMRA}, the following is true:  if $f$ is a generic $K$-sparse
vector with respect to the orthonormal basis ${\mathcal V}$, the only vectors
in the $H$-orbit of $f$ which are also $K$-sparse are of the form
$e^{\I \alpha} f$. 

\rev{Although the $H$-orbit of $f$ is a real algebraic subvariety of $V$
 containing $\{e^{\I \alpha} f \}$, 
we know of no general algebraic geometry result which can be used to analyze
  when a generic linear subspace of $V$ will intersect the orbit $Hf$
  exactly in $\{e^{\I \alpha}f\}$}.
To prove \rev{our} result, we will actually prove something stronger. Rather than consider the $H$-orbit of a vector $f$, we will consider the linear span of its orbit and
prove that the only $K$-sparse vectors in the linear span of the orbit
of $f$ are scalar multiples of $f$. The advantage of working with the linear span is that we can use techniques from linear algebra to understand when a
linear subspace (the linear span of our orbit) intersects the $\binom{N}{K}$
$K$-dimensional linear subspaces consisting of vectors which are $K$-sparse with respect to the given orthonormal basis ${\mathcal V}$.

The price we pay for working with the linear span of an orbit instead of
directly working with the orbit is that if the $H$ orbit
of $f$ has real dimension $M$, then  its linear span is a complex
linear subspace of complex dimension $M$ or equivalently real dimension $2M$
(see Proposition \ref{prop.linearspan}).
As a result, the sparseness bound we obtain may not be optimal. However, when $\dim H << \dim V$, as is the case in cryo-EM, this gap is not significant.

Finally, we remark that for the general MRA problem with group $G$~\eqref{eq:mra}, 
we can at best recover the $G$-orbit of a vector $f$ from its moments. However,
by imposing the prior condition that the vector is sparse with respect to a given
basis we have the possibility of recovering a vector up to a global phase. The reason is that for a general orthonormal basis ${\mathcal V}$ of $V$, the sparse vectors are not invariant under the action of $G$.

\subsection{Proof of Theorem \ref{thm.mainMRA} \label{sec:proof_main_thm}}
Let $H$ be a group acting on a vector space $V$ and $f \in V$
any vector.
We denote by  $\Lf$
the {\em linear span} of the $H$-orbit $Hf$. By definition $\Lf = \left\{ \sum a_i (h_i \cdot f)| a_i \in \C, h_i \in H\right\}$ and it is the smallest linear subspace
containing the orbit $Hf$.

Let $V = \oplus_{\ell = 1}^L V_\ell^{R_\ell}$ be a unitary
representation of a
compact group $G$ and let 
$H= \prod_{\ell =1}^{L} U(N_\ell)$.
Given a vector $f \in V$, we can write $f = \sum_{\ell = 1}^L \sum_{r = 1}^{R_\ell}f_\ell[r]$, 
where $f_\ell[r]$ is in the $r$-th copy of the irreducible representation $V_\ell$.
As above,  we can view our vector $f$ as an $L$-tuple
$(A_1, \ldots , A_{L})$ of 
$N_\ell \times R_\ell$ 
matrices with
$A_\ell=\left (f_\ell[1])^T, \ldots (f_{\ell}[R_\ell])^T\right)$. Viewing
$V_\ell^{R_\ell}$ as the vector space of $N_\ell \times R_\ell$ matrices, 
the linear span $\Lf$ of the orbit $Hf$ is the product
of the \rev{the linear spans of the} $U(N_\ell)$ orbits of the matrices $A_\ell$, where elements of $U(N_\ell)$ acts
on $A_\ell$ by left multiplication.

\begin{proposition} \label{prop.linearspan}
	Let $V= \oplus_\ell^L V_\ell^{R_\ell}$ be a unitary
	representation of a compact group $G$ and
	let $H = \prod_{\ell=1}^L U(N_\ell)$. If $f \in V$ is
	represented by an $L$-tuple $(A_1, \ldots , A_L)$
	of $N_\ell \times R_\ell$ matrices, 
	then $$\dim_\C \Lf = \sum_{\ell = 1}^L (\rank A_\ell) N_\ell,$$
	where   $\dim_\C$ denotes the dimension of $\Lf$ as a complex vector space.
	In particular,
	$$\dim_\C \Lf \leq \sum_{\ell =1}^L M_\ell,$$
	where $M_\ell = \min (N_\ell R_\ell, N_\ell^2)$.
\end{proposition}
\begin{proof}
	Since the linear span of the $H$-orbit of $f = (A_1, \ldots , A_L)$ is the product
	of the linear spans of the $U(N_\ell)$-orbits of the matrices $A_\ell$, 
	it suffices to prove that the linear span of the $U(N_\ell)$-orbit
	of the matrix $A_\ell$ in $V_\ell^{R_\ell}$ has dimension
	$(\rank A_\ell) N_\ell$.
	
	Let $r_\ell = \rank A_\ell$ and to simplify notation
	assume that the first $r_\ell$ columns $f_\ell[1]^T, \ldots , f_\ell[r_\ell]^T$
	of $A_\ell$ are linearly independent. Since $\rank A_\ell = r_\ell$, for $r > r_\ell$  there are unique
	scalars $b_{1,r}, \ldots , b_{r_\ell,r}$ such that $f_\ell[r] = \sum_{i=1}^{r_\ell}
	b_{i,r} f_\ell[i]$.
	
	Let $\L_{A_\ell}$ be the $r_\ell N_\ell$-dimensional linear subspace of $V_{\ell}^{R_L}$ consisting
	of $N_\ell \times R_\ell$ matrices $B$ such that for $r>r_\ell$,
	$B_r = \sum_{i=1}^{r_\ell}
	b_{i,r} B_i$, where $B_i$ denotes the $i$-th column of the matrix $B$.
	Since $U(N_\ell)$ acts linearly, the linear relations on the columns of $A_\ell$ are preserved by the action of
	$U(N_\ell)$, so the linear span of $U(N_\ell)A_\ell$ lies in the subspace
	$\L_{A_\ell}$. Conversely, we note that the linear span of $U(N_\ell)A_\ell$ contains
	the open set $\L_{A_\ell}^o$  of $\L_{A_\ell}$, parameterizing matrices whose
	first $r_\ell$ columns are linearly independent. The reason this holds
	is because any invertible $N_\ell \times N_\ell$ matrix is a linear combination
	of unitary matrices and any element of $\L_{A_\ell}^o$ can be obtained by applying
	some invertible matrix to $A_\ell$.
\end{proof}

\rev{
\begin{remark}
	Note that the real dimension of the $U(N_\ell)$-orbit of the
	matrix $A_\ell$ considered in the proof of Proposition \ref{prop.linearspan}
	has real dimension $r_\ell N_\ell$. It follows that for any vector
	$f \in V$,  $\dim_\C \Lf = \dim_\R Hf$. In particular, 
	the real dimension of $\Lf$ is twice the real dimension of the orbit $Hf$.
\end{remark}
}

\rev{ To prove the theorem, we need to show that the 
set $\mathcal{U}$ of orthonormal bases $\V$, such that for every
subset $S \subset [1,N]$ with
$|S| =K$ and with $K \leq  M = \sum_{\ell=1}^L \min(N_\ell R_\ell, N_\ell^2)$
the following statement hold.
\begin{enumerate}
\item[(1)] For generic $f \in \L_S(\V), \L_f \cap \L_S(\V)$ is the line
  spanned by $f$.
\item[(2)] For generic $f \in \L_S(\V)$, if $|S'| = K$ and $S' \neq S$ then $\L_f \cap \L_S(\V) = \{0\}$.
\end{enumerate}
%contains a Zariski open set in $O(N)$ (resp. $U(N)$).

For a fixed subset $S$ with $|S| = K$, let ${\mathcal U}_S$
be the set of orthonormal bases such that (1) and (2) hold for $S$. Then, 
${\mathcal U} = \cap_{S} {\mathcal U}_S$. Since the intersection
of a finite number of Zariski open sets is Zariski open, it suffices
to prove that each ${\mathcal U}_S$ contains a Zariski open set.
Moreover, the proof is identical up to indexing for each subset $S$ so
we will assume, for simplicity of notation, that $S= \{1, \ldots , K\}$.

Given a vector $f \in V$, let ${\mathcal B}_f$ be the
set of orthonormal bases such that $f \in \L_{\{1 , \ldots , K\}}$. Note
  that ${\mathcal B}_f$ is a Zariski closed subset of
  $O(N)$ (resp. $U(N)$) defined by the equation
  $f_1 \wedge \ldots \wedge f_K \wedge f = 0$, where
  $f_1, \ldots , f_K$ are the first $K$ vectors of an ordered basis.
}

\begin{proposition} \label{prop.crucial}
	Let $f\in V$ be any non-zero vector and let $\Lf$ be the linear span
	of its orbit under $H$. 
%	Let   $\L_{\{1, \ldots , K\}}(\V)$ denote the linear subspace
%	spanned by the first $K$ vectors of the ordered basis $\V$,
        \rev{Let $M = \sum_{\ell=1}^L \min(N_\ell R_\ell, N_\ell^2)$.}
	Then, if $K \leq  N-M$, \rev{for the generic orthonormal basis
          $\V \in {\mathcal B}_f$:}
	\begin{enumerate}
		\item[(1)] $\Lf$ intersects $\L_{\{1, \ldots , K\}}(\V)$ in the line spanned by $f$; 
		\item[(2)] $\L_f \cap \L_S(\V) = \{0\}$ if $S \neq \{1, \ldots , K\}$.
	\end{enumerate}
\end{proposition}
\rev{The set of orthonormal bases ${\mathcal V} \in {\mathcal B}_f$ for which
  conditions (1) and (2) of Proposition~\ref{prop.crucial}
  are not satisfied is defined by polynomial
  equations. This means that the set of bases satisfying (1) and (2)
  is Zariski open, and to prove Proposition~\ref{prop.crucial}
  we just need to show that this set is non-empty; i.e., we just need
  to show that there exists a basis ${\mathcal V} \in {\mathcal B}_f$
  which satisfies conditions (1) and (2).
To do this we need to introduce some notation and prove a lemma.
}
%To prove Proposition \ref{prop.crucial} we need to introduce some notation and prove a lemma.

Fix an orthonormal basis $e_1, \ldots , e_N$ of a Hermitian vector space $V$
of dimension $N$. For $S \subset [1,N]$ with $|S| = K$, let $\L_S = \Span \{e_i\}_{i \in S}$ and $\L_S^*$ be the open subset of $\L_S$ of vectors whose expansion
with respect to the basis $\{e_i\}_{i \in S}$ have all non-zero coordinates.
In other words, $\L_S^* = \L_S \setminus (\bigcup_{S' \neq S} \L_{S'})$.

For a given vector $w \in V$, let $\Gr_w(M,V)$ be the subvariety of the Grassmannian of $M$-dimensional linear subspaces of $V$ that contain $w$.
\begin{lemma} \label{lem.technical}
	If $K \leq N-M$, then for any vector $w \in \L^*_{\{1, \ldots ,K\}}$
	the generic $M$-dimensional linear subspace $\L \in \Gr_w(M,V)$ satisfies
	the following conditions:
	\begin{enumerate}
		\item $\L \cap \L_{\{1, \ldots , K\}}$ is the line spanned by $w$; 
		\item  $\L \cap \L_S = \{0\}$ for $S \neq \{1, \ldots , K\}$ and $|S| = K$.
	\end{enumerate}
\end{lemma}
\begin{proof}[Proof of Lemma \ref{lem.technical}]
	The subset of $\Gr_w (M,V)$ parameterizing linear subspaces
	intersecting $\L_{\{1, \ldots , K\}}$ in dimension greater than one is locally
	defined by a polynomial equation and therefore a proper algebraic subset.
	Likewise, for any $S \neq \{1, \ldots , K\}$ the subset of $\Gr_w(M,V)$
	parameterizing linear subspace $\L$ such that $\L \cap \L_S \neq \{0\}$
	is also defined by a polynomial equation, and thus is again a proper algebraic subset.
	In particular, the set of $\L \in \Gr_w(M,V)$ which do not satisfy conditions
	(1) and (2) lie in a proper algebraic subset of $\Gr_w(M,V)$. Therefore, the generic subspace $\L \in \Gr_w(M,V)$ satisfies conditions (1) and (2).
\end{proof}
\begin{proof}[Proof of Proposition \ref{prop.crucial}]
	Choose a fixed orthonormal basis $\{e_1, \ldots , e_N\}$ and
	let $(\L, w)$ be an $M$-dimensional linear subspace and vector satisfying
	the conclusions (1) and (2)
	%\tb{conditions?}
	of Lemma~\ref{lem.technical}. If we choose $w$ so that
	$|w| = |f|$, then we can find a rotation $g \in U(N)$ such that
	$g\cdot (\L,w) = (\L_f, f)$. The orthonormal basis $\{v_i = g \cdot e_i\}_{i=1, \ldots N}$
	satisfies conditions (1) and (2) of the proposition.
\end{proof}

%We can conclude the proof of Theorem~\ref{thm.mainMRA}
%with the following proposition.
\begin{proposition} \label{prop.onetogeneric}
	Let ${\mathcal V}$ be an ordered orthonormal basis for $V$ and assume
	that there is a non-zero vector
	$f_0 \in L_{\{1, \ldots , K\}}(\V)$ such that $\dim \L_{f_0} \cap
	\L_{\{1, \ldots K\}}(\V) =1 $ and $\L_{f_0} \cap \L_{S}(\V)  = \{0\}$ for $S \neq \{1, \ldots , K\}$. Then, for a generic $f \in \L_{\{1, \ldots , K\}}(\V)$
	the same property holds.
\end{proposition}
\begin{proof}
	Given an orthonormal basis ${\mathcal V}$, the set
	$D$ of  $f \in \L_{\{1, \ldots , K\}}(\V)$
	%\tb{I don't think we used this notation before.}
	which satisfy the condition that $\dim \Lf \cap \L_{\{1, \ldots K\}}(\V)  > 1$ or
	$\dim \Lf \cap \L_S(\V) > 0$ for $S \neq \{1, \ldots , K\}$ is defined by polynomial
	equations. By hypothesis, we know that $D \neq \L_{\{1, \ldots , K\}}$ since $f_0 \notin D$ so
	its complement is necessarily Zariski dense.
\end{proof}
\rev{At this point we have proved the following. For a fixed vector
non-zero $f_0 \in V$, there is a Zariski open set
$\mathcal{U}_{f_{0}}\subset {\mathcal B}_{f_0}$ such that for every $\V
\in \mathcal{U}_{f_0}$ the generic vector $f \in \L_{\{1,\ldots ,
K\}}$ satisfies conditions (1) and (2) of Proposition
\ref{prop.crucial}.
}
\rev{To complete the proof, we observe that the set of all bases
is $\bigcup_{f_0 \in \P(V)} \mathcal{B}_{f_0}$
and our desired set of bases
is $\bigcup_{f_0 \in \P(V)} \mathcal{U}_{f_0}$,
where $\P(V)$ is the
projective
space of lines in $V$.
This set is open in $O(N)$ (resp. $U(N)$)
because it is the complement of the projection to the first of the Zariski closed
set $Z = \{(\V, f)| \V \in (\mathcal{B}_{f} \setminus \mathcal{U}_f)\}
\subset O(N) \times \P(V)$
(resp. $U(N) \times \P(V)$) and this projection is proper (meaning
it takes Zariski closed sets to Zariski closed sets) because the projective
space $\P(V)$ is a proper variety.
}
\subsection{Examples} \label{sec:examples_sparsity}

\subsubsection{MRA with rotated images model}
Using Theorem~\ref{thm.mainMRA} we can obtain sparsity bounds for recovering a
generic image from its second moment as in Section \ref{ex.rotatedimage}.

Recall that in this model the Fourier transform of an image is represented
as a radially discretized band-limited function on $\C^2$,
and the function $f$ is determined by \rev{an $L$}-tuple $(A_{-L'}, \ldots , A_L')$
vectors in $\C^R$, where \rev{$L'=(L-1)/2$} is the bandlimit and $R$ is the number of radial
samples. The ambiguity group is $H = (S^1)^{2L+1}$.
In the notation of Theorem~\ref{thm.mainMRA},
we have $M_\ell  = 1$ for $\ell = -L', \ldots, L'$.
In particular, for any vector $f \in V$,
$\dim \Lf \leq \rev{L}$. Hence, by Theorem~\ref{thm.mainMRA} we can conclude
that if
$K \leq \dim V - \rev{L}$, then for a generic orthonormal basis, 
a generic $K$-sparse vector can be recovered from its second moment.
Since $\dim V = R\rev{L}$,
if the number of radial samples $R \geq 2$, then
the sparsity level required for signal recovery is $K\leq \frac{R-1}{R}N$, namely,
linear in $\dim V$.
If only one radial sample is taken ($R=1$), then this problem reduces to the  MRA model on the circle,  
which is equivalent to the Fourier phase retrieval problem~\cite{bendory2017fourier}.

\subsubsection{Sparsity bounds for two-dimensional tomography from
	unknown random projections}
Following the model of Section~\ref{sec.2dtomography}, the
unknown image $f$ is viewed as a $\rev{L}\times R$ matrix~$A$
and the projected second moment determines the matrix $A^*A$.
Thus, the ambiguity group is $U(L)$ (complex images) or \rev{$O(L)$} (real images).
The orbit of a generic signal $f$ has dimension $M$, where
$M = \min ( \dim V, \rev{L}^2)$. Since $\dim V = \rev{L}R$, 
we have $M = \min (R\rev{L}, \rev{L^2})$.
In order to be able to recover
sparse signals, we need to take $R > \rev{L}$; i.e., the number of
radial samples must exceed the number of frequencies. Specifically, 
Theorem~\ref{thm.mainMRA} implies that for a generic ordered orthonormal basis
$\V$ we 
can recover $K$-sparse signals
where $\rev{K = (R-L)L}$. In particular, if
$R > p\rev{L}$ with
$p > 1$, then a generic $K$-sparse signal is uniquely
determined by its second moment if $K \leq {{p-1}\over{p}}N$, 
where $N = \dim V$.

\section{Application to cryo-EM} \label{sec:cryoEM}

This section is devoted to the application of the results of Section~\ref{sec:second_moment} and Section~\ref{sec:sparsity} to single-particle cryo-EM: the main motivation of this work.

Recent technological breakthroughs in cryo-EM have sparked a revolution in structure biology---the field that studies the structure and dynamics of biological molecules---by
recovering an abundance of new molecular structures at near-atomic resolution. 
In particular, cryo-EM allows recovering molecules that were notoriously difficult to
crystallize (e.g., different types of membrane proteins), the sample preparation procedure is significantly simpler (compared to alternative technologies) and preserves the molecules in a near-physiological state, and it allows reconstructing multiple functional states. 

In this section, we describe the mathematical model of cryo-EM in detail, formulate the ambiguities of recovering the three-dimensional structure from the second moment, and then derive the sparsity level that allows resolving these ambiguities based on Theorem~\ref{thm.mainMRA}.

\subsection{Mathematical model} \label{sec:cryoEM_model}
Let $L^2(\R^3)$ be Hilbert space of complex valued $L^2$ functions
on~$\R^3$. The action of $SO(3)$ on $\R^3$ induces a corresponding
action on $L^2(\R^3)$, which we view as an
infinite-dimensional representation of $SO(3)$.
In cryo-EM we are interested in the action of $SO(3)$ on
the subspace of $L^2(\R^3)$ corresponding to the Fourier 
transforms of real valued functions on $\R^3$,  representing  the coulomb  potential of an unknown molecular structure.

Using spherical coordinates $(r, \theta, \phi)$ we consider a finite dimensional
approximation of 
$L^2(\R^3)$ by discretizing $f(r, \theta, \phi)$ 
with $R$ samples $r_1, \ldots , r_{R}$, of the radial coordinates
and bandlimiting the corresponding spherical functions $f(r_i, \theta, \phi)$. 
This is a standard assumption in the cryo-EM literature, see for example~\cite{bandeira2020non}.
Mathematically, this means that we approximate the infinite-dimensional representation $L^2(\R^3)$ with
the finite dimensional representation
$V = (\oplus_{\ell =0}^L V_\ell)^R$, where
$L$ is the bandlimit, and $V_{\ell}$ is the $(2\ell +1)$-dimensional
irreducible representation of $SO(3)$, corresponding to harmonic
polynomials of frequency $\ell$.
An orthonormal basis for $V_\ell$
is the set of spherical harmonic polynomials $\{Y_\ell^m(\theta,
\phi)\}_{m = -\ell}^\ell$. We use the notation $Y_\ell^m[r]$
to consider the corresponding spherical harmonic as a basis vector for
functions on the $r$-th spherical shell. The dimension of this
representation is $R(L^2 + 2L+1)$. 
%\tb{How did you get to $L^2+2L$? Because $\sum_{\ell=0}^{L}(2\ell+1)=L^2 + 2L+1$.}

Viewing an element of $V$ as a radially discretized function on $\R^3$, 
we can view
$f \in V$ as an $R$-tuple
$$f = (f[1], \ldots , f[R]),$$
where
$f[r] \in L^2(S^2)$ is an $L$-bandlimited function.
Each $f[r]$ can be expanded in terms of the basis functions $Y_\ell^m(\theta, \varphi)$ as follows 
\begin{equation} \label{eq.function}
	f[r] = \sum_{ \ell=0}^L\sum_{m=-\ell}^\ell A_{\ell}^m[r]
	Y_{\ell}^m.
\end{equation}
Therefore, the problem of determining a structure reduces to determining the unknown coefficients $A_\ell^m[r]$ in \eqref{eq.function}.

Note that when
$f$ is the Fourier transform of a real valued function,  the coefficients
$A_\ell^m[r]$ are real for even $\ell$ and purely imaginary for odd $\ell$~\cite{bhamre2015orthogonal}.

\subsection{The second moment of the cryo-EM model} \label{sec.2ndmoment}
In this section, we first formulate the second moment of the MRA model~\eqref{eq:mra} for $G=$SO(3) and functions of the form~\eqref{eq.function}. Then, we show that this is equivalent to the second moment of the cryo-EM model (Lemma~\ref{lem:2ndmoment_cryo-EM}) and derive the ambiguity group of this model (Corollary~\ref{cor:2ndmoment_cryo-EM}).

Consider the MRA model with $G=$SO(3) and functions of the form~\eqref{eq.function}.
Using the expansion from the previous section and
the functional representation of the second moment~\eqref{eq.functional_second_moment}, we can write
\begin{equation} \label{eq.cryosecondmoment}
	m^2_f = \sum_{r_1,r_2=1}^{R} \sum_{\ell = 0}^{L} \left(\sum_{m = -\ell}^{\ell}
	A_{\ell}^m[r_1] \overline{A_{\ell}^m[r_2]}\right) \sum_{m'=-\ell}^{\ell}
	Y_{\ell}^{m'}[r_1]
	\overline{Y_{\ell}^{m'}[r_2]},
\end{equation}
where the notation $Y_\ell^m[r]$ denotes the corresponding spherical harmonic
in the $r$-th copy of $V_\ell \subset L^2(S^2)$.
To simplify notation,  set
\begin{equation}
	B_{\ell}[r_1,r_2] = \sum_{m=-\ell}^\ell \rev{A_{\ell}^m}[r_1] \overline{\rev{A_{\ell}^m}[r_2]}.
\end{equation}
This 
can be viewed as an inner product of the coefficient vector
$\left(A_{\ell}^{-\ell}[r_1] , \ldots , A_{\ell}^{\ell}[r_1]\right)$
from the $r_1$-shell and the coefficient vector 
$\left(A_{\ell}^{-\ell}[r_2] , \ldots , A_{\ell}^{\ell}[r_2]\right)$
from the $r_2$ shell.
Let $A_{\ell}\in\C^{(2\ell+1)\times R}$  and $B_\ell\in\C^{R\times R}$ be matrices consisting of the coefficients 
$$A_{\ell} = \left(A_{\ell}^m[r_i]\right)_{m=-\ell, \ldots , \ell , i = 1, \ldots R},$$
and
$$B_{\ell}  = \left(B_\ell[r_i,r_j]\right)_{i,j=1\ldots,R}.$$
Then, the second moment determines the matrices
\begin{equation} \label{eq:Bl}
	B_\ell = A_\ell^T A_\ell, \quad  \ell = 0, \ldots L.
\end{equation}
Remarkably, \rev{unlike the tomographic projection $\R^2 \to \R^1$,}
the tomographic projection operator~\eqref{eq:tomographic_projection} does not affect the second moment for \rev{$\SO(3)$}. Therefore, in the context of the second moment, we can treat cryo-EM as a special case of the MRA model\rev{~\eqref{eq:mra}}, where $G$ is the group of three-dimensional rotations SO(3) and $V$ is a a discretization  of $L^2(\R^3)$ as in~\eqref{eq.function}.
This fact has \rev{been} recognized (implicitly) already by Zvi Kam~\cite{kam1980reconstruction}.
For completeness, we prove the following lemma.
\begin{lemma} \label{lem:2ndmoment_cryo-EM}
	Assume a function of the form~\eqref{eq.function}. Then, the second moment of the cryo-EM model~\eqref{eq:cryoEM_model} is the same as the second moment of the MRA model~\eqref{eq:mra} with $G=$SO(3). Namely, the tomographic projection operator in~\eqref{eq:cryoEM_model} does not affect the second moment.
\end{lemma}
\begin{proof}
	Consider the projected second moment of a function $f \in V$ for fixed $(r_1,r_2)$:
	\begin{align} \int_{SO(3)} T(g\cdot f[r_1](\theta_1, \phi_2))
		T(\overline{g\cdot f[r_2](\theta_2, \phi_2)}\;dg = & (T \times T)\int_{SO(3)} (g\cdot f[r_1](\theta_1, \phi_1) \overline{(g\cdot f[r_2])(\theta_2, \phi_2})dg\nonumber\\
		& = (T\times T)(m^2_f[r_1,r_2](\theta_1, \phi_1, \theta_2, \phi_2))\\
		& =\sum_{\ell = 0}^{L} B_\ell[r_1,r_2]
		\sum_{m=-\ell}^{\ell} Y_{\ell}^m(\pi/2, \varphi_1)[r_1]
		\overline{Y_{\ell}^m(\pi/2, \varphi_2)[r_2]}\nonumber.
	\end{align}
	Here, $T \times T$ is the product of tomographic projections so
	$(T \times T)f(\theta_1, \phi_1, \theta_2, \phi_2) = f(\pi/2, \phi_1,\pi/2,\phi_2)$. Note that the first equality holds because the linear operator $T$
	commutes with integration over the group $SO(3)$.
	Let $P_\ell$ be the Legendre polynomial of degree $\ell$.
	Since, up to constants~\cite[Section 2.2]{atkinson2012spherical},
	\begin{equation} \label{eq.legendre} \sum_{m=-\ell}^\ell Y^m_\ell(\pi/2, \varphi_1)
		\overline{Y^m_\ell(\pi/2, \varphi)}
		= P_\ell(\cos(\varphi_1 -\varphi_2)),
	\end{equation}
	we have
	\begin{equation} \label{eq.crux}
		\int_G T(g\cdot f[r_1]) \overline{ T(g\cdot f[r_2])}\;dg =		\sum_{\ell = 0}^L B_\ell[r_1,r_2] P_\ell(\cos(\varphi_1 - \varphi_2)).
	\end{equation}
	Since the Legendre polynomials are orthonormal functions of $\varphi = \varphi_1 - \varphi_2$, we can determine the coefficients $B_\ell[r_1, r_2]$ from \eqref{eq.crux}.	
	Thus we can conclude that no information is lost from the taking the projected second moment.
\end{proof}

\begin{corollary} \label{cor:2ndmoment_cryo-EM}
	Assume a  function of the form
	\begin{equation*}
		f[r] = \sum_{ \ell=0}^L\sum_{m=-\ell}^\ell A_{\ell}^m[r]
		Y_{\ell}^m.
	\end{equation*}
	Then, the second moment of the cryo-EM model~\eqref{eq:cryoEM_model} is given by~\eqref{eq:Bl}. Therefore, the second 
	moment
	determines the coefficient matrices~$A_\ell, \, \ell=0,\ldots,L$ up to the action of the ambiguity group
	$\prod_{\ell =0}^{L} U(2\ell +1)$.
	Moreover, if we consider functions $f[r]$ which are the Fourier
	transforms of real-valued functions on $\R^3$ (which is the  scenario in cryo-EM), then the coefficients
	$A_\ell^m[r]$ are real for even $\ell$ and purely imaginary for odd $\ell$~\cite{bhamre2015orthogonal}, 
	and the ambiguity group is $\prod_{\ell =0}^{L} O(2\ell +1)$.
\end{corollary}

%\begin{remark}
%  As noted above, in cryo-EM we consider functions $f \in L^2(\R^3)$ which are the Fourier
%  transforms of real-valued functions on $\R^3$ so
%  that the coefficients
%  $A_\ell^m[r]$ are real for even $\ell$ and purely imaginary for odd $\ell$~\cite{bhamre2015orthogonal}, 
%  and the ambiguity group is $\prod_{\ell =0}^{L} O(2\ell +1)$.
%\end{remark}

\subsection{Recovery of sparse structures from the second moment\label{sec:recovery_sparse}}
Based on Theorem~\ref{thm.mainMRA}, we now prove that in cryo-EM a $K$-sparse signal
can be recovered from the second moment when $K \lessapprox  N/3$.

\begin{theorem} \label{thm.mainCRYO}
	Assume a function of the form~\eqref{eq.function}, where the number of shells satisfies $R \geq 2L+1$.
	Let $V = \oplus_{\ell =0}^{L} V_{\ell}^{R}$
	and let $N = \dim V$.
	Then, if 
	\begin{equation*}
		\frac{K}{N} \leq \frac{2/3 L^3 + L^2 + L/3}{2L^3 + 5L^2 + 4L +1}\approx\frac{1}{3},
	\end{equation*}  
	then
	for a generic choice of orthonormal basis ${\mathcal V}$, 
	a generic $K$-sparse function $f \in V$
	is uniquely determined by its second moment, up to a global phase.
\end{theorem}

\begin{proof}
	The dimension of the representation $V$ is $N=R(L+1)^2$. Thus, if $R \geq 2L +1$, then $N = \dim V \geq 2L^3 + 5L^2 + 4L +1$.
	On the other hand,
	since $R \geq \dim V_\ell$ for all $\ell$, we know by Proposition~\ref{prop.linearspan} that
	for $f \in V$ the linear span $\L_f$ of the orbit of $f$ under the ambiguity
	group $\oplus_{\ell =0}^L O(2\ell +1)$ has dimension at most
	$$\sum_{\ell =0}^L  (2\ell +1)^2 = 4/3 L^3 + 4L^2 + 11L/3 +1 .$$
	Therefore, by Theorem \ref{thm.mainMRA}, if $K \leq 2/3 L^3 + L^2 + L/3$ then for a generic choice of orthonormal basis, 
	a generic $K$-sparse vector $f \in V$ is uniquely determined by
	its second moment. 
	\end{proof}
	
	\begin{corollary}
		Under the conditions of Theorem~\ref{thm.mainCRYO}, a three-dimensional structure of the form~\eqref{eq.function} can be recovered from   $n$ realization from the cryo-EM model when $n=\omega(\sigma^4)$.
	\end{corollary}

\rev{	
	\begin{remark}[Near-optimality]
		While the sparsity level of Theorem~\ref{thm.mainCRYO} is not necessarily optimal, it is optimal up to a constant. Thus,  we say that our sparsity bound is near-optimal.
	\end{remark}
	\begin{remark}
		A recent paper~\cite{bendory2022autocorrelation} showed that a three-dimensional structure composed of a finite number of ideal point masses (or its convolution with a fixed kernel with a non-vanishing Fourier transform) can be recovered from the second moment. Theorem~\ref{thm.mainCRYO} is far more general as it includes sparse structures under almost any basis.
		Yet,~\cite{bendory2022autocorrelation} also suggests an algorithm which harnesses sparsity in the wavelet domain, for which our result does not necessarily hold (since Theorem~\ref{thm.mainCRYO} holds for generic bases and we cannot verify that any particular basis satisfies the generic condition).  
	\end{remark}
	
	\begin{remark}[Spherical-Bessel expansion]
		Our analysis assumes a model of multiple shells as in~\eqref{eq.function}. However, a similar analysis can be carried out to related models, such as spherical-Bessel expansion, where the coefficients $A_{\ell}^m[r]$ are expanded by 
		\begin{equation*}
			A_{\ell}^m[r] = \sum_{s=1}^{S_\ell} \tilde{A}_{\ell}^m[s] j_{\ell,s}[r],
		\end{equation*}
		where the $j_{\ell,s}[r]$ are the normalized spherical Bessel functions. 
		The ``bandlimit'' $S_\ell$ is determined by a sampling criterion, akin to Nyquist sampling criterion~\cite{bhamre2017anisotropic}.
		This expansion has been useful in various cryo-EM tasks, see for example~\cite{levin20183d,bendory2018toward,bendory2022sparse}.
		Our analysis can be applied to molecular structures represented using the spherical-Bessel expansion, where the only difference is the way we count the dimension of the representation. 
	\end{remark}
}
	
	\section{Crystallographic phase retrieval} \label{sec:phase_retrieval}
	
	The crystallographic phase retrieval problem is the problem of recovering a signal in $\R^N$ or $\C^N$ from its power spectrum. As seen from Section~\ref{ex.powerspectrum}, this is equivalent to recovering a signal from its second moment for the action of either the cyclic group $\Z_N$ or the dihedral group.  
	However, because each irreducible representation
	appears exactly once, 
	\rev{Theorem~\ref{thm.mainMRA} provides an uninformative bound of $K\leq 0$.} 
	%we cannot use Theorem \ref{thm.mainMRA} to obtain a sparsity bound which ensures generic recovery of signals.
	
	In~\cite{bendory2020toward}, the authors conjectured that when
	$\R^N$ is given by the standard basis,   a generic $K$-sparse
	vector in $\R^N$ can be recovered, up to unavoidable ambiguities, from its power
	spectrum if $K \leq N/2$ \rev{if the support is not an arithmetic progression}. This conjecture was proved for a few specific cases but a complete  proof of this conjecture is  beyond current techniques.
	In~\cite{ghosh2022sparse}, it was shown that for large enough $N$,  $K$-sparse, symmetric signals are determined uniquely from their power spectrum for \rev{$K=O(N/\log^5N)$}.

\rev{On the other hand, for generic bases, the following provable
  optimal bound for phase retrieval was recently obtained
  \cite{edidin2023generic} using the techniques of this paper. Unlike
  the conjectures of~\cite{bendory2020toward}, this result makes no
  assumption on the support of the signal with respect to the given
  basis.
%	Here, we rely on previous results on phase retrieval for generic frames and derive the following bound for recovery of vectors contained in a linear subspace.
\begin{theorem} \cite[Theorem 1.1]{edidin2023generic} \label{thm.phaseretrieval}
Let ${\mathcal V}$ be a generic basis for $\R^N$. Then, if 
$K \leq N/2$, a generic $K$-sparse vector can be recovered
from its power spectrum, up to a global phase.
\end{theorem}
}        %\begin{remark} 	
%	\begin{proof}
%		Let $\L$ be a $K$-dimensional linear subspace of $\C^N$.
%		Let $f_1, \ldots , f_K\in \C^N$ be a basis for $\L$,
%		and let $h_i = Ff_i \in \C^N$, where $F$ is the $N \times N$ discrete Fourier transform. The vectors $h_1, \ldots , h_K \in \C^N$ determine an $N$-element
%		frame $q_1, \ldots q_N$ on $\C^K$, where $q_j = (h_{1j}, \ldots , h_{Kj})$.
%		If we expand a vector
%		$f \in \L$ as $f = \sum a_i f_i$ with $a_i \in \K$, then the power spectrum
%		of $f$ is exactly the collection of phaseless frame measurements
%		$(|\langle a,q_1 \rangle|^2, \ldots, |\langle a, q_N \rangle|^2)%$, 
%		where $a= (a_1, \ldots , a_K) \in \C^K$.
%		
%		Since the linear subspace $\L$ is arbitrary and the discrete Fourier
%		transform is an isomorphism, for a generic choice
%		of $\L$ the frame $q_1, \ldots, q_N$ is generic. By \cite[Theore%m 3.4]{balan2006signal}, if $N \geq 2K$ the generic vector $a \in \C^K$ can be recovered
%		from these measurements. Therefore, the generic vector $f \in \L%$
%		can be recovered from its power spectrum.
%	\end{proof}
	
%\tb{I'm not sure I understand this remark.}
%	\rev{\begin{remark}
%		Note that Theorem~\ref{thm.phaseretrieval} is weaker than Theorem~\ref{thm.mainMRA}, our corresponding results for generic sparsity conditions in the MRA problem, because it is only a statement for a
%		single
%		generic linear subspace as a opposed to a generic sparsity condi%tion. In future work, we will pursue an extension of Theorem~\ref{thm.phaseretrieval} to generic sparsity conditions.
%	\end{remark}}

	\section{Discussion and future work} \label{sec:discussion}
	In this paper, we have derived general sparsity conditions under which the sample complexity of the MRA model~\eqref{eq:mra} is only $n=\omega(\sigma^4)$ rather than 
	$n=\omega(\sigma^6)$ in the general case. 
	We have further applied the result to cryo-EM, showing that if a molecular structure can be represented with $\sim N/3$ coefficients in a generic basis, then the sample complexity is quadratic in the variance of the noise.
	Next, we delineate a few possible extensions of these results.
	
	%\paragraph{Non-compact groups}
	\paragraph{\rev{Linear transformations which are not compact groups}}
	
	Our MRA model~\eqref{eq:mra} assumes a compact group. 
	However, in some important situations, the group is non-compact, for instance, the group of rigid motions SE(d)~\cite{bendory2022compactification}.
	One challenge of working with non-compact groups is that their representations do not necessarily decompose into a sum of irreducibles which makes the representation-theoretic analysis of the second moment more difficult.
	\rev{The problem is even more challenging when considering a combination of a group action with a general linear operator; this is for example the situation when considering sub-pixel measurements~\cite{bendory2022super}.}
	
	\paragraph{Sample complexity for specific bases}
	Our main theoretical result, Theorem~\ref{thm.mainMRA}, holds for almost all bases but it is very difficult to say if they hold for a specific basis, such as wavelets, since the algebraic conditions \rev{on the bases} are implicit.
	An important future work is to derive conditions for recovery from the second moment for specific bases,  and ideally for all bases. 
	(In~\cite{bendory2022sparse}, recovery from the second moment  of structures composed of ideal point masses  was proven.)
	%(In~\cite{bendory2022sparse}, it was shown that signal recovery is possible from the second moment for sparse structures under the canonical basis.)
	
%        \dan{This paragraph is a little inaccurate and will need to be rewritten.}
	\paragraph{Unified theoretical framework with phase retrieval}
	In Section~\ref{sec:phase_retrieval}, Theorem~\ref{thm.phaseretrieval}, we \rev{discussed} sparsity conditions for recovering a signal from its power spectrum, which is the second moment of the simplest MRA model, where a signal in $\R^N$ is acted upon by $\Z_N.$
	This problem is called the phase retrieval problem.
	We wish to consolidate the proof techniques of Theorem~\ref{thm.mainMRA} and those used to prove Theorem~\ref{thm.phaseretrieval} \rev{in} one general theoretical framework, \rev{which should yield optimal dimension bounds for
recovering a signal from its second moment.} 

	\paragraph{Multi-target detection}
	The multi-target detection model was devised 
	to design a new computational \rev{paradigm} for recovering small molecular structure using cryo-EM~\cite{bendory2019multi}.
	Without delving into the technical details,
	the second moment of this model is provided by the 
	diagonals of the matrices $B_\ell, \, \ell=0,\ldots,L,$ that describe the second moment of the cryo-EM model~\eqref{eq:Bl}~\cite{bendory2018toward}. 
	Deriving the conditions for signal recovery from these diagonals will have important implications to the sample complexity of the multi-target detection model and to understanding the fundamental limitations of the   cryo-EM technology.

	\paragraph{Alternative priors} This work shows that the sample complexity of MRA and cryo-EM can be significantly improved if the signal can be sparsely represented. 
	An interesting future research thread is studying alternative priors that can improve the sample complexity, such as  statistical priors, data-driven priors (e.g., based on AlphaFold~\cite{jumper2021highly}), \rev{semi-algebraic priors~\cite{dym2022low}}, or priors based on the statistical properties of proteins~\cite{wilson1949probability,singer2021wilson}. 

%	\dan{Do we want a paragraph to propose the problem of understanding information lost by projected moments?}
%	\tb{I think we have enough... }
	
\section*{Acknowledgments}
We thank Nicolas Boumal for his notes on~\cite{kam1980reconstruction}, and Guy Sharon and Oscar Mickelin  for helping, respectively, with Figure~\ref{fig:cryo} and Figure~\ref{fig:wavelet}.
This research is supported by the BSF grant no. 2020159. T.B. is also supported in part by the NSF-BSF grant no. 2019752, and the ISF grant no. 1924/21 and D.E. was also supported by NSF-DMS 1906725 and NSF-DMS 2205626.

\appendix

\section{Representation theory}
\label{sec:representation_theory}

\subsection{Terminology on representations}
Let $G$ be a group. A  (complex) representation of $G$ is a homomorphism, 
$G \stackrel{\pi} \to \GL(V)$, where $V$ is a complex vector space and $\GL(V)$ is the group
of invertible linear transformations $V \to V$. Given a representation
of a group $G$, we can define an action of $G$ on $V$ by $g \cdot v = \pi(g)v$.
Since $\pi(g)$ is a linear transformation, the action of $G$ is necessarily linear, meaning that for any vectors $v_1, v_2$ and scalars $\lambda, \mu \in \C$
$g \cdot ( \lambda v_1 + \mu v_2) = \lambda(g \cdot v_1) + \mu (g \cdot v_2)$.
Conversely, given a linear action of $G$ on a vector space $V$, we obtain
a homomorphism $G \to \GL(V)$, $g \mapsto T_g$, where
$T_g \colon V \to V$ is the linear transformation $T_g(v) = (g \cdot v)$.
Thus, giving a representation of $G$ is equivalent to giving a linear action
of  $G$ on a vector space $V$. Given this equivalence, we will follow
standard terminology and refer to a vector space $V$ with a linear action
of $G$ as a {\em representation of $G$}.

A representation $V$ of $G$ is {\em finite dimensional} if $\dim V < \infty$.
In this case, a choice of basis for $V$ identifies $\GL(V) = \GL(N)$,
where $N = \dim V$. Given a Hermitian inner product $\langle\cdot,\cdot \rangle$
on $V$,  we say that a representation is {\em unitary} if for any two vectors
$v_1, v_2 \in V$
$\langle g \cdot v_1,   g \cdot v_2 \rangle = \langle v_1,   v_2 \rangle$.
If we choose an orthonormal basis for $V$, then the representation
of $G$ is unitary if and only if the image of $G$ under the homomorphism
$G \to GL(N)$ lies in the subgroup $U(N)$ of unitary matrices.

A representation $V$ of a group $G$ is {\em irreducible} if $V$ contains
no non-zero proper $G$-invariant subspaces.

\subsection{Representations of compact groups}
Any compact group $G$ has a $G$-invariant measure called a Haar measure.
The Haar measure $dg$ is typically normalized so that $\int_G dg =1$.
If $V$ is a finite-dimensional representation of a compact group
and $(\cdot,\cdot)$ is any Hermitian inner product, then
the inner product $\langle \cdot,\cdot\rangle$ defined by the formula
$\langle v_1, v_2 \rangle = \int_G (g\cdot v_1, g \cdot v_2)\;dg$ is $G$-invariant. As a consequence we obtain the following fact.
\begin{proposition}
	Every finite dimensional representation of a compact group is unitary.
\end{proposition}
Using the invariant inner product we can then obtain the following decomposition
theorem for finite dimensional representations of compact group.
\begin{proposition}
	Any finite dimensional representation of a compact group
	decomposes into a direct sum of irreducible representations.
\end{proposition}

\rev{
	\begin{example}
		Most representations that naturally occur in imaging
		and signal processing are not irreducible. For example, consider the action
		of the cyclic group $\Z_N$ on $\C^N$ by cyclic shifts; i.e., if
		$T$ is the generator of $\Z_N$, then $T \cdot (x_0, \ldots , x_{N-1}) =
		(x_1, x_2, \ldots, x_{N-1}, x_0)$. 
		To see that this representation is reducible, 
		let $e_0, \ldots , e_{N-1}$ be the standard basis and take $\omega = e^{2 \pi \iota/N}$. 
		With this notation, if  $0 \leq n \leq N-1$, the one-dimensional subspace $V_n$ spanned by
		the vector $f_n = e_0 + \omega^{n}e_1 + \omega^{2n}e_2 + \ldots + \omega^{(n-1)n}e_{n-1}$ is invariant under $T$ because $T \cdot f_n = \omega^n f_n$.
		The vectors $f_0, \ldots , f_{N-1}$ are the Fourier basis for $\C^N$ and
		$\C^N$ decomposes as the sum of one-dimensional representations $V_0 \oplus \ldots \oplus V_{N-1}$. In general, if $G$ is an abelian compact group
		then any complex representation of $G$ will decompose into a sum of one-dimensional representations.

		For non-abelian groups, or even real representations
		of abelian groups, this need not be the case. If we consider the action of $\Z_N$ on $\R^N$ by cyclic shifts,  
		then $\R^N$ decomposes into a sum of one and two-dimensional irreducible
		representations.
		For example if $N=4$, then $\R^4$ decomposes as the sum
		$V_0 \oplus  V_1 \oplus  V_2$, where $V_0 = \Span (e_0 + e_1 + e_2 +e_3)$,
		$V_1 = \Span(e_0 - e_2, e_1 -e_3)$ and $V_2 = \Span (e_0 -e_1 + e_2 - e_3)$.
	\end{example}
}

If $V$ is a representation, then $V^G = \{v \in V| g \cdot v = v\}$
is a subspace which is called the subspace of invariants.
%\tb{I thought that an invariant subspace is a subspace that satisfies that if $v\in V$ then $gv\in V$ for any $g\in G.$}

\subsection{Schur's Lemma}
A key property of irreducible unitary representations is Schur's Lemma.
Recall that  a linear transformation $\Phi $ is $G$-invariant if $g \cdot \Phi v = \Phi g \cdot v$.

\begin{lemma} \label{lem.schur}
	Let $\Phi \colon V_1 \to  V_2$ be a $G$-invariant linear transformation
	of finite dimensional irreducible representations
	of a group $G$ (not necessarily compact).
	Then, $\Phi$ is either zero or an isomorphism.
	Moreover, if $V$ is a finite dimensional irreducible unitary representation
	of a group $G$ then any $G$-invariant linear transformation
	$\phi \colon V \to V$ is multiplication by a scalar.
\end{lemma}

\subsection{ Dual, $\Hom$ and tensor products of representations} \label{sec.adjoint}
If $V_1$ and $V_2$ are representations of a group $G$, 
then the vector space $\Hom(V_1, V_2)$ of linear transformations
$V_1 \to V_2$ has a natural linear action of $G$ given by
the formula $(g\cdot A)(v_1) = g \cdot A(g^{-1}v_1)$.
In particular, if $V$ is a representation of $G$, then $V^* = \Hom(V, \C)$
has a natural action of $G$ given by the formula $(g \cdot f)(v) = f(g^{-1}v)$.

A choice of inner product on $V$ determines an identification of vector spaces
$V =V^*$, given by the formula $v \mapsto \langle \cdot , v \rangle$.
%\tb{did you mean $v \mapsto \langle \cdot , v \rangle$?}.
If $V$ is a
unitary representation of $G$ then with the identification
of $V=V^*$ the dual action of $G$ on $V$ is given by the formula
$g \cdot_* v = \overline{g} \cdot v$.
Likewise, if $V_1$ and $V_2$ are two representations
then we can define an action of $G$ on $V_1 \otimes V_2$ by the formula
$g\cdot (v_1 \otimes v_2) = (g \cdot v_1) \otimes (g \cdot v_2)$.

Given two representations spaces $V_1, V_2$ there is an isomorphism of representations
$V_1 \otimes V_2^* \to \Hom(V_2, V_1)$ given by 
the formula $v_1 \otimes f_2 \mapsto \phi$, where the linear transform
$\phi \colon V_2 \to V_1$
is defined by the formula
$\phi(v_2) = f_2(v_2) v_1$.
In particular, we can identify $V \otimes V^*$ with $\Hom(V,V)$.

\rev{
	\section{Grassmannians}
	The set $\Gr(M,V)$ of $M$-dimensional linear subspaces of an $N$-dimensional
	vector space $V$ has the natural structure as a projective manifold, called
	the {\em Grassmannian of $M$ planes in $V$}. The Grassmannian $\Gr(M,V)$
	has dimension $M(N-M)$ and there are a number of ways to see the manifold
	structure and compute the dimension.
	
	Given an ordered basis $v_1, \ldots , v_M$
	for an $M$-dimensional linear subspace
	$\Lambda$, we can associate a full rank $M \times N$ matrix
	$A_\Lambda = [v_1 \ldots v_M]$. Conversely, given a full rank $M \times N$
	matrix~$A$, the columns of $A$ determine an ordered basis for an $M$-dimensional linear
	subspace of $V$. Two matrices $A,A'$ correspond to the same linear subspace
	if an only there is an invertible $M \times M$ matrix $P$ such
	that $A = PA'$. Hence, the Grassmannian can be identified as
	the quotient $F(M,N)/\GL(M)$, where $F(M,N)$ is the set of full rank
	$M \times N$ matrices. Since $F(M,N)$ has dimension $MN$ and $\GL(M)$
	has dimension $M^2$, the quotient
	has dimension $MN - M^2 = M(N-M)$.
	
	To see that $\Gr(M,V)$ is a projective variety, we note that
	if $v_1, \ldots , v_M$ and $v'_1, \ldots , v'_M$ are two bases
	for an $M$-dimensional linear subspace $\Lambda$ then,  $v_1 \wedge \ldots \wedge v_M  = \lambda (v'_1 \wedge \ldots
	\wedge v'_M)$  for some scalar $\lambda$. Here $\wedge$ denotes the exterior
	product. Thus there is a well-defined
	map $\Gr(M,V) \to \mathbb{P}(\bigwedge^M V)$ which sends
	the point representing the linear subspace $\Lambda$ to the
	exterior product $v_1 \wedge \ldots \wedge v_M$ where $v_1, \ldots , v_M$
	is any basis for $\Lambda$. Moreover, this map is an embedding
	because $v_1 \wedge \ldots \wedge v_M = \lambda(v'_1 \wedge \ldots \wedge
	v'_M)$ for some scalar $\lambda$ if and only if $v_1, \ldots , v_M$
	and $v'_1, \ldots , v'_M$ span the same $M$-dimensional linear subspace.
	%\tb{What is the $\wedge$ notation?}
	This is embedding is called the Pl\"ucker embedding.
	
	For more on Grassmannians, see \cite[Chapter 1, Section 5]{griffiths1978principles} or \cite[Lecture 6]{harris1995algebraic}.
}

\bibliographystyle{siamplain}
\bibliography{ref}

\end{document}